\documentclass[runningheads]{llncs}
\usepackage[T1]{fontenc}
\usepackage{graphicx}
\usepackage{caption}
\usepackage{subcaption}
\usepackage{todo}
\usepackage{amsmath}
\usepackage{amssymb}
\usepackage{amsfonts}
\usepackage{xspace}
\usepackage{xcolor}
\usepackage{hyperref}
\usepackage{thmtools} 
\usepackage{thm-restate}

\author{Zeta Avarikioti\inst{1} \and
Paweł Kędzior\inst{2} \and
Tomasz Lizurej\inst{3} \and
Tomasz Michalak\inst{4}}
\institute{Department of Informatics, TU Wien, Austria \and
University of Warsaw, Poland \and
NASK and University of Warsaw, Poland \and
IDEAS NCBR and University of Warsaw, Poland}

\authorrunning{Avarikioti et al.}

\graphicspath{ {./images/} }

\begin{document}
\newcommand{\prob}{\mathbb{P}}

\newcommand{\adv}{{V_{\mathcal{A}}}}
\newcommand{\cg}{{G_C}}

\newcommand{\game}{\Gamma_F}

\newcommand{\state}{\mathcal{S}}
\newcommand{\states}{\mathbb{S}}

\newcommand{\block}{\mathsf{block}}
\newcommand{\base}{B}
\newcommand{\winner}{W}

\newcommand{\reward}{\mathsf{reward}}

\newcommand{\players}{N}
\newcommand{\player}{N}

\newcommand{\party}{\mathcal{P}}

\newcommand{\actions}{\Omega}

\newcommand{\expectedutil}{\boldsymbol{\mathbb{E}}}

 \newcommand{\notenamem}[2]{{\textcolor{red}{\footnotesize{\bf (#1:} {#2}{\bf ) }}}}
 \newcommand{\za}[1]{{\notenamem{Zeta}{#1}}}

 \newcommand{\tl}[1]{{\notenamem{Tomasz}{#1}}}

\newcommand{\decision}{\mathsf{decision}}
\newcommand{\abandon}{\mathsf{abandon}}
\newcommand{\continue}{\mathsf{continue}}

\newcommand{\attack}{\textit{Bribe \& Fork}\xspace}

\newtheorem{conj}{Conjecture}
\title{\attack: 
\\Cheap PCN Bribing Attacks via Forking Threat
}
\titlerunning{\attack}
%

%
\maketitle

\begin{abstract}
In this work, we reexamine the vulnerability of Payment Channel Networks (PCNs) to bribing attacks, where an adversary incentivizes blockchain miners to deliberately ignore a specific transaction to undermine the punishment mechanism of PCNs. 
While previous studies have posited a prohibitive cost for such attacks,  we show that this cost can be dramatically reduced (to approximately \$125), thereby increasing the likelihood of these attacks. 
To this end, we introduce \attack, a modified bribing attack that leverages the threat of a so-called feather fork which we analyze with a novel formal model for the mining game with forking. 
We empirically analyze historical data of some real-world blockchain implementations to evaluate the scale of this cost reduction. 
Our findings shed more light on the potential vulnerability of PCNs and highlight the need for robust solutions.
\keywords{Blockchain, Payment Channels Networks, Timelock Bribing, Feather Forking}
\end{abstract}
\section{Introduction}
The financial world was transformed by blockchains such as Bitcoin~\cite{nakamoto2008bitcoin} and Ethereum~\cite{wood2014ethereum}. 
While blockchains offer a number of benefits, their scalability remains a significant challenge when compared to traditional centralized payment systems~\cite{croman2016scaling}. One promising solution to this issue is the so-called \emph{payment channel networks (PCNs)} that move most of the transaction workload off-chain~\cite{layer2}.




Several PCN proposals~\cite{aumayr2020generalized,aumayr2021bitcoin,aumayr2022sleepy,avarikioti2021brick,avarikioti2020cerberus,decker2019eltoo,decker2015fast,dziembowski2019perun,dziembowski2018general,jourenko2020lightweight,poon2016bitcoin,spilman2013anti} have been laid forward so far, each design offering some unique combination of features.
Nonetheless, the core idea behind payment channels remains the same, that is to facilitate off-chain transactions among parties connected, either directly or indirectly via a path, on a network operating on top of the blockchain layer, the PCN.

To participate in a PCN, two parties can lock funds in a joint account on-chain, thereby opening a payment channel.  Subsequently, the parties can transact off-chain by simply updating (and signing) the distribution of their funds. When either party wants to settle the account, or in other words close the payment channel with its counterparty, they can publish the last agreed distribution of the channel's funds. However, each update constitutes a valid closure of the channel from the perspective of the blockchain miners. 
As a result,  a malicious party may publish an outdated update to close the channel holding more than it currently owns. 
To secure the funds against such attacks, payment channels enforce a \emph{dispute period}. During this period, the funds remain locked to allow the counterparty to punish any malicious behavior, and if so, claim all the funds locked in the channel. 

Hence, the security of PCNs, like the most widely deployed Bitcoin Lighting Network~\cite{poon2016bitcoin},  crucially relies on financial incentives.
Specifically, during the dispute period, the punishment mechanism should enforce that a malicious party is always penalized and an honest party should never lose its funds.
Unfortunately, this is not always the case, as argued by Khabbazian et al.~\cite{timelocked} and Avarikioti et al.~\cite{zeta}. For instance, Lighting channels are susceptible to the so-called \emph{timelock bribing attacks}. In such an attack, a malicious party posts an old update transaction on-chain, attempting to close its channel with more funds than it presently possesses. Concurrently, the party bribes the miners to ignore the punishment transaction of the counterparty for the duration of the dispute period. This bribe is typically offered to the miners in the form of high transaction fees.

Naturally, the success of this attack depends on the value of the bribe. Avarikioti et al.~\cite{zeta} showed that 
a bribing attack will be successful if the bribe is no smaller than: 
$\dfrac{f_1-f}{\lambda_{min}},$
\noindent where $\lambda_{min}$ is the fraction of the mining power controlled by the least significant miner in the underlying proof-of-work blockchain,
$f_1$ is the sum of the fees of a block containing the punishment transaction, and $f$ is the sum of the fees of a block containing average transactions.
%
We observe, however, that as $\lambda_{min}$ can be arbitrarily small, the bribing amounts required can significantly exceed the funds typically locked in PCNs, rendering the bribing attack impractical.
Moreover, even in blockchains with rather concentrated mining power, like in Bitcoin~\cite{long2022measuring}, the cost of a bribing attack is very high. 
For example, conservatively assuming that the smallest miner has $10^{-4}$ of the total mining power\footnote{In the Section~\ref{sec:experimental}, we experimentally show, that the mining power of the weakest miner in the system can be fairly assumed to be of magnitude $10^{-12}$}, the cost of the attack as analyzed in~\cite{zeta} would be at least $1$ BTC, for $f_1-f \approx 10^{-4}$ BTC. As a result, it would be irrational to perform a bribing attack of this sort, as the average closing price for $1$ BTC between, for instance, 2019 and 2022 was $23,530.92$ USD\footnote{\url{statmuse.com}},  which is more than 10-fold the current total value locked on average in a Lighting channel\footnote{\url{https://1ml.com/statistics}}.
This naturally leads to questioning \emph{whether there is potential to amplify such attacks to the extent they pose a genuine threat to the security of PCNs}.

\subsubsection*{Our Contribution}
In this work, we show that \emph{bribing costs can be significantly reduced}, thereby {making timelock bribing attacks a realistic threat}.
We do so by extending the bribing attack to leverage not only the structure of transaction fees but also a threat to fork the blockchain, known as a \emph{feather fork attack}~\cite{magnani2018feather,shalini2019survey}.
In our context, a given miner executes a feather fork attack by announcing a self-penalty transaction $tx_p$. Whenever the self-penalty transaction appears on the blockchain, the miner is incentivized to fork the punishment transaction $tx_1$ on the blockchain, i.e., the miner will try to extend the blockchain based on the predecessor of the block $tx_1$ including some other block.
Specifically, a feather-forking miner is bribed to commit collateral, betting that their fork will win the race. As a consequence, their threat of forking becomes considerably credible, incentivizing other miners to follow their fork. 
The collateral is of a similar magnitude as the bribe in~\cite{zeta}, however, the miners only lock it \emph{temporarily}. We call our attack \attack.
With the feather fork at hand, the bribing cost may now be reduced from $\frac{f_1-f}{\lambda_{min}}$ to approximately:
$$\frac{2\overline{f}+2(f_1-f)}{\lambda_s},$$ 
\noindent where $\overline{f}$ is the average fee of a single transaction, and $\lambda_s$ is the mining fraction of the most significant miner. Recall that $f_1$ and $f$ denote the sum of fees of a block containing the punishment transaction and only average transactions respectively. 
To demonstrate the cost reduction of \attack, we reexamine the previous example for Bitcoin, with  $\overline{f} \approx 10^{-4}$ BTC, $f_1-f \approx \overline{f}$, and $\lambda_s \approx 20\%$.
Now, $\lambda_s$ replaces in the denominator the previously presented $\lambda_{min} << 10^{-4}$, thus \emph{yielding a bribe at least 1000 times smaller than the one  in~\cite{zeta}}.

To derive this result, we present a formal model of mining games with forking, extending the conditionally timelocked games introduced in~\cite{zeta}. In the game with forks, miners may now choose, in each round, (a)  which transactions to mine, and (b) whether they want to continue one of the existing chains or they intend to fork one of the chains. All miners know the choices of the winner of each block, as a feather-forking miner locks collateral on-chain.

To empirically estimate the cost reduction of \attack, we analyze the historical data of real-world blockchain implementations. Among others, we analyze the average block rewards and fees, as well as the hash power present in the system and available to a single miner, primarily for Bitcoin in 2022. Given the officially available data, we observe that \emph{the cost of our attack can be as cheap as \$125} (for 1 BTC $\approx$ \$25.000). In general, the cost of our attack can be up to  $ 10^{10}$ times cheaper than the bribe required in~\cite{zeta} according to our findings. Hence, even considering a collateral of around $\$30,000$, \attack is substantially more cost-efficient, and, by extension, more probable to occur.

\section{Background}\label{subsec:mining}
In this section, we first describe the necessary context required to understand \attack.

\subsection{Timelocked Bribing Attack}
Whenever miners decide to create a new block, they select some set of transactions from all transactions posted on the mempool, which is a database of all publicly visible transactions.
Mining pools and individual miners usually choose the transactions with the highest fees first, as they are part of their reward for a successfully mined block. Miners are aware that some transactions may be dependent on each other. For instance, two transactions that spend the same Unspent Transaction Output (or UTXO) in Bitcoin, cannot be both published on-chain; the transaction that is validated first, i.e., is included on a block of the longest chain, immediately deems invalid the other transaction. If from two interdependent transactions,  only one can be published directly, while the other is timelocked and thus can only be published after some time elapses, we refer to this pair of transactions as a \emph{conditionally timelocked output}.
This conditionally timelocked output is the target of timelock bribing attacks: the owner of the coins of the transaction that is valid only after the timelock expires attempts to bribe the miners to ignore the currently valid competing transaction. 
Thus, for the miners that observe transactions on the mempool, sometimes it may be profitable to censor one transaction in order to mine another one that provides a greater gain in the future.

\subsection{Timelock Bribing in the Bitcoin Lightning Network}
Next, we describe the timelock bribing attack in the context of the most widely deployed payment channel network, namely the Bitcoin Lightning Network (LN).
In LN, a single on-chain transaction called the funding transaction, opens a channel between parties $\party_1, \party_2$. Next, parties exchange with each other signed messages off-chain which update the state of their accounts. If the parties are honest and responsive, they may close the channel in collaboration. To do so, the parties post a mutually signed transaction that spends the output of the funding transaction and awards each party their fair share of funds.  
However, if a dishonest party $\party_2$ attempts to publish on-chain an old state that she profits from comparably to the latest agreed state, her funds will remain locked for the so-called dispute period. During this period, 
the other (rational) party $\party_1$ will try to revoke the old state, by sending a transaction $tx_1$ to the mempool called the revocation transaction (or breach remedy).  Transaction $tx_1$ awards all the channel funds to the cheated party  $\party_1$. We denote by $f_1$ the miner's fee to include a block with $tx_1$.

In this case, the malicious party $\party_2$ can launch a timelocked bribing attack,  attempting to bribe the miners to ignore $tx_1$ for the dispute period $T$ such that $\party_2$ gains the additional funds. Specifically, $\party_2$ may send in the mempool a block with fee $f_2$ that includes transaction $tx_2$, with  $f_2>f_1$,  that is only valid if no miner includes in their winning block containing $tx_1$ within time $T$. 
Consequently, if the revocation transaction $tx_1$ is not published on-chain within $T$, $\party_2$ can spend the funds of the old state and the next winning miner will be awarded $f_2$. 
The pair of transactions $tx_1$ and $tx_2$ is now a conditionally timelocked output.


Assuming that for an average block of transactions, the users get in total $f$ fees,  the following holds~\cite{zeta}: if $f_2-f >\frac{f_1-f}{\lambda_{min}}$ then all rational miners will choose to wait for $T$ rounds and publish a block containing $tx_2$. 

\subsection{Feather Fork Attack}
A feather fork, as introduced in~\cite{forking_discussion}, is an attack on Bitcoin wherein a miner threatens to fork the chain if selected transactions are included. This intimidation mechanism aims to subtly alter the miners' block acceptance policy: 
the threatened miners may exclude the selected transactions in order to mitigate the risk of losing their mining reward~\cite{temporary}. As feather forking relies on economic incentives, the attacker may increase their probability of success by bribing other miners to follow their short-lived fork, e.g., committing to pay them the block rewards they may lose by censoring the selected transactions. 

Unlike a ``hard'' fork, where miners exclusively mine their own chain version regardless of its length compared to other versions, a feather fork entails mining on the longest chain that excludes selected transactions and does not fall significantly behind its alternatives. Thus, a feather fork is less disruptive and more likely to be adopted by the network, making it a potentially powerful tool in the hands of a malicious actor. 
In this work, we employ this tool to enhance the likelihood of a successful timelocked bribing attack.

\section{\attack Attack}
\label{sec:self:penalty}
We now introduce our novel attack, termed \attack, that combines the timelock bribing attack in LN with the feather fork attack. 
We assume the existence of a payment channel between parties $\party_1$ and $\party_2$. Similarly to timelock bribing, we consider $\party_2$ to be malicious and attempt to close the channel with $\party_1$ in an old state using the transaction $tx_2$. Consequently, $\party_1$ is expected to attempt to revoke the old state. To prevent the inclusion of the revocation transaction, $\party_2$ bribes the miner $\player_s$ with the highest mining power $\lambda_s$ to threaten others with a fork if they add the unwanted transaction $txs_1$ on-chain. 
This action is implemented through a self-penalty mechanism where the bribed miner temporarily locks collateral which can only be reclaimed in case a block $txs_2$ containing $tx_2$ appears on-chain (and thus any block $txs_1$ containing $tx_1$ is not included on-chain). In essence, the bribed miner ``bets'' that the revocation transaction will be censored, thus rendering the threat credible for the rest of the miners. 

To realize \attack, there are two mechanisms that should be implementable on-chain: a) the bribe transaction that should only be spendable if $txs_2$ is included on-chain, and b) the self-penalty mechanism that enables the bribed miner $\player_s$ to lock collateral $P$ (with transaction $tx_{p_1}$) and then reclaim it only if $txs_2$ is included on-chain (with transaction  $tx_{p_2}$). 
We implement the bribing and self-penalty mechanisms in Bitcoin script, using the conditioning enabled by the UTXO (Unspent Transaction Output) structure.
We note that in Ethereum,  preparing a smart contract that has access to the state of the closing channel suffices to implement the bribing and self-penalty mechanisms.

In detail, a single bribing transaction $tx_b$ and two special transactions $tx_{p_1}, tx_{p_2}$ are introduced. Let us assume that the cheating party $\party_2$ is bribing the miners to launch the attack conditioned on the inclusion of its transaction $tx_2$. To do so, $\party_2$ creates a bribe transaction $tx_b$ with input the party's money from $tx_2$ and outputs three UTXOs, one given to the miner that mines this transaction, one dummy output owned by player $\player_s$ (i.e., the miner bribed to perform feather forking), and one that returns the rest of the money to $\party_2$.
Now, $\player_s$ creates a transaction $tx_{p_1}$, locking a deposit $P$, that is spendable via a multisignature of $m$-out-of-$n$ parties (e.g., $m=n/2$), one of which is $\player_s$'s signature.
%
Then, $tx_{p_2}$ is created with two inputs: the output of $tx_{p_1}$ and the dummy output of $tx_b$. Consequently, $tx_{p_2}$ is spendable only if it is signed by at least $m$ parties of the predefined set $n$ and $tx_2$ is validated on-chain.
Upon receiving $tx_{p_2}$  signed by at least $m-1$ parties, $\player_s$ signs and posts it on-chain.
Assuming that no subset of size $m-1$ of the rest $n-1$ parties will collude with the miner to spend the deposit, the deposit can be claimed by the miner only if transaction $tx_2$ is included on-chain. The security of this scheme depends on the selection of the $n-1$ parties, which can be in principle conditioned on the honest majority assumption of the blockchain via subsampling.\footnote{One could also consider using  Trusted Execution Environment (TEE) that outputs $tx_{p_2}$ only when $txs_2$ appears on the blockchain. Additionally, note that on Ethereum,  preparing a smart contract that has access to the state of the closing channel is sufficient to implement the penalty mechanism.}

Now, the malicious party $\party_2$ together with a selected miner $\player_s$ can launch the $\attack$ attack as follows: 
They can create special transactions $tx_b$ (as a companion transaction for $tx_2$) and the self-penalty transactions $tx_{p_1}, tx_{p_2}$ to publicly announce a credible threat that the transaction $tx_1$ will be forked once it appears on the blockchain. They publish these special transactions in the mempool as transaction sets $txs_2$ (that contains $tx_2$ and $tx_b$) and $txs_{p_1}, txs_{p_2}$ (that contain $tx_{p_1}, tx_{p_2}$, respectively).

We show later that this attack significantly reduces the cost of the required bribe from  $\frac{f_1-f}{\lambda_{min}}$ to approximately $\frac{2\overline{f}+2(f_1-f)}{\lambda_s}$. The required collateral that is eventually reclaimed by the bribed miner is expected to be $\lambda_s \cdot \base$ (where $\base$ is a constant in Bitcoin block reward independent of the user fees), which is comparable to the original bribe needed in \cite{zeta}.
Figure~\ref{fig:cases} illustrates the comparison of \attack with \cite{zeta}, while Figure~\ref{fig:attack:details} below depicts the details of \attack.

\begin{figure}[!ht]
\centering
\begin{subfigure}[][250pt][t]{.32\textwidth}
  \centering
  \caption{}
  \includegraphics[width=\textwidth]
  {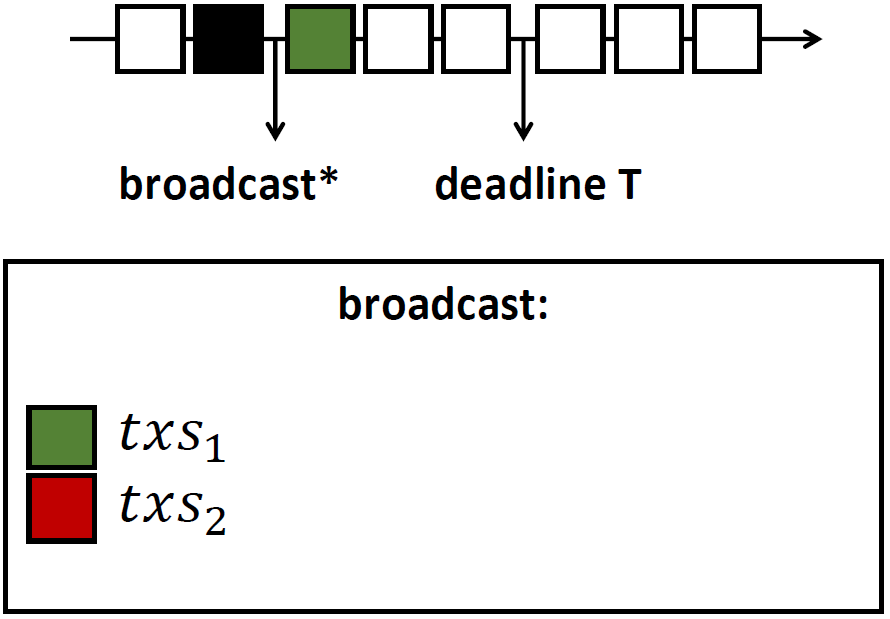}
  \label{fig:normal}
\end{subfigure}
\begin{subfigure}[][250pt][t]{.32\textwidth}
  \centering
  \caption{}
  \includegraphics[width=\textwidth]
  {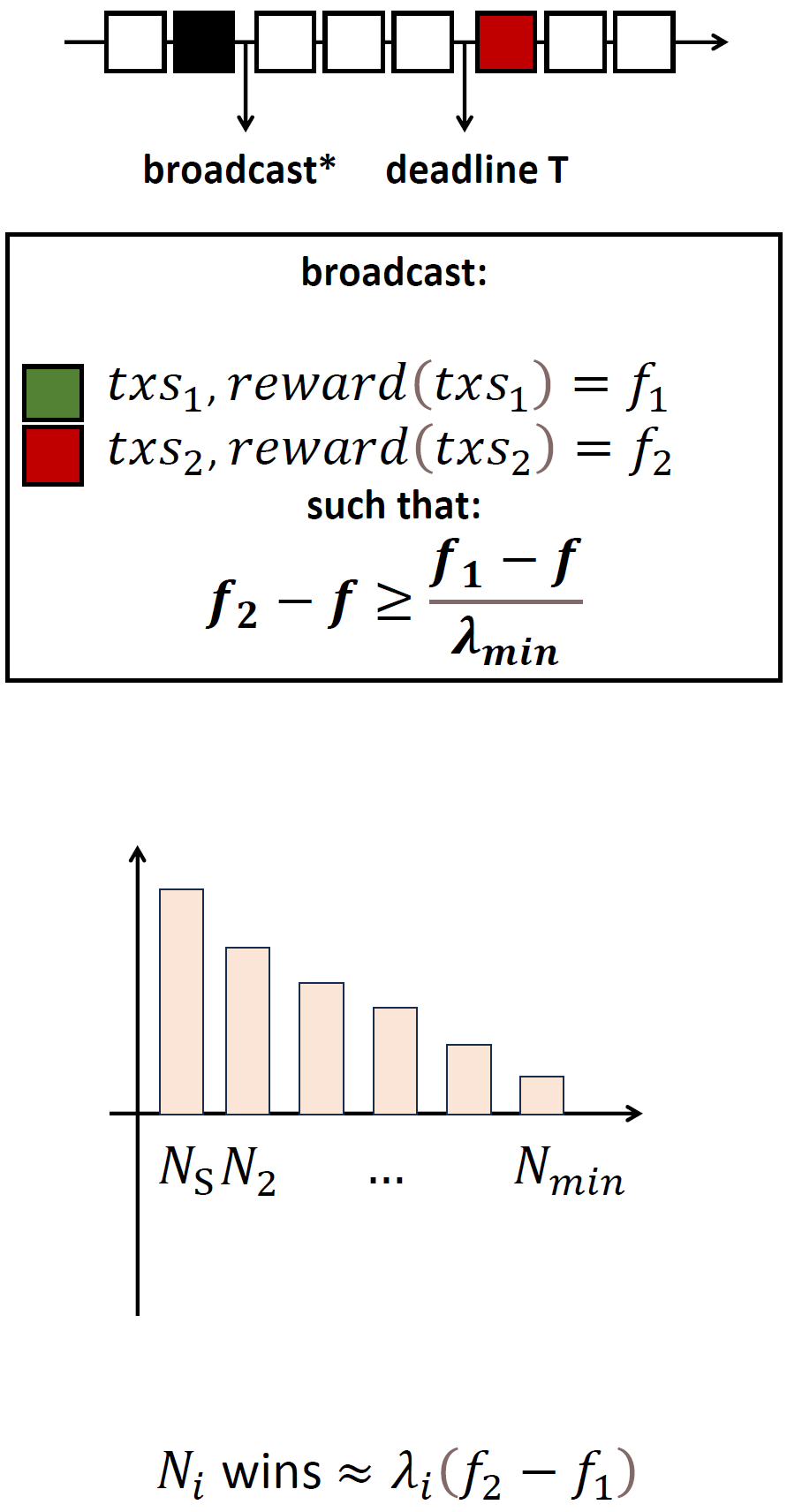}
  \label{fig:adv}
\end{subfigure}
\begin{subfigure}[][250pt][t]{.32\textwidth}
  \centering
  \caption{}
  \includegraphics[width=\textwidth]
  {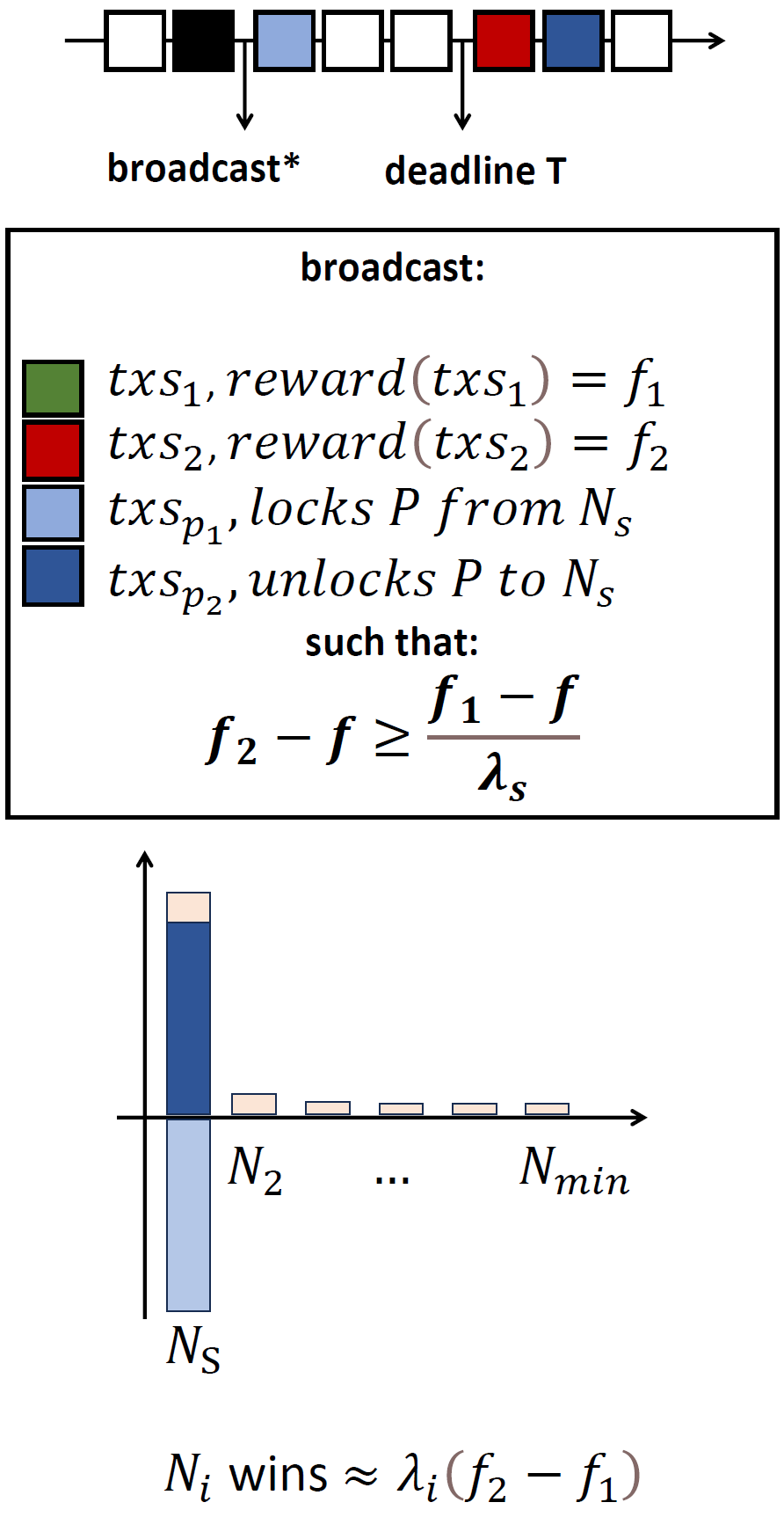}
  \label{fig:adv2}
\end{subfigure}
\vspace{2.1cm}
\caption{Comparison of the honest execution, the attack from~\cite{zeta} and our $\attack$. \textbf{(a) Honest execution:}  once an old state appears on-chain (black rectangle), $\party_1$ gets an option to revoke this state with a transaction $tx_1$ (included in the block $txs_1$ which is published in the first round). \textbf{(b) Attack  in~\cite{zeta}:} the bribing party publishes $tx_2$ and $tx_b$ included in a single block $txs_2$, with a large miner fee (reversely proportional to the fraction of the mining power $\lambda_{min}$ of the least significant miner). The miners skip mining $txs_1$ in the first round, and mine $txs_2$ in the last round.  \textbf{(c) \attack:} the bribing party publishes $txs_2$ with a fee sufficient to bribe only the strongest miner (with $f_2-f$ reversely proportional to $\lambda_s$). The strongest miner publishes the self-penalty transactions $tx_{p_1}, tx_{p_2}$ that can be mined in transaction sets $txs_{p_1}, txs_{p_2}$. In the first round, the miner $\player_s$ locks $P \approx \lambda_s \cdot \base$ to the deposit transaction $tx_{p_1}$, thus threatening other miners that they will be forked once $txs_1$ is mined before the deadline. After the deadline the transaction set $txs_2$ is published and the miner $\player_s$ may collect back the deposit using $txs_{p_2}$.}
\label{fig:cases}
\end{figure}

\subsubsection*{Implementation Details of $\attack$}
\label{sec:details}
Figure~\ref{fig:attack:details} contains a diagram depicting the details of our attack.
At first, a Lightning Channel is opened with a single funding transaction and allows parties $\party_1, \party_2$ to make an arbitrary number of off-chain state transitions of their funds. Once one of the parties ($\party_2$ in our example) decides to publish an old state ($comm_i$ in our example) at a chain of length $T_0$, the opportunity to manipulate the behaviour of miners' is opened. 
Before the chain reaches length~$T_0 + 1$, the transactions $tx_1$ and $tx_2$ are published. Transaction $tx_1$ allows $\party_1$ to revoke a dishonestly committed state. Transaction $tx_2$ allows $\party_2$ to collect and manage dishonest funds after $T$ rounds.
Along with $tx_1$ and $tx_2$, the bribing transaction $tx_{b}$ and self-penalty transactions $tx_{p_1}, tx_{p_2}$ are published. 
On the chain of length $T_0$, the miners may decide to mine the transaction $tx_{p_1}$ that would lock some amount of coins of one of the miners (say $\player_s$).
If so, all miners are threatened that they will be forked once $tx_1$ appears on the blockchain until the chain reaches length~$T_0 + T$.
On the chain of length~$T_0+T$, the transaction $tx_2$ along with bribing transaction $tx_b$ may be published and the selected miner $\player_s$ may collect back the deposit with the transaction $tx_{p_2}$.
\label{sec:attack:details}
\begin{figure}[]
    \centering    
\makebox[\textwidth][c]{\includegraphics[angle=0,origin=c,width=20cm]{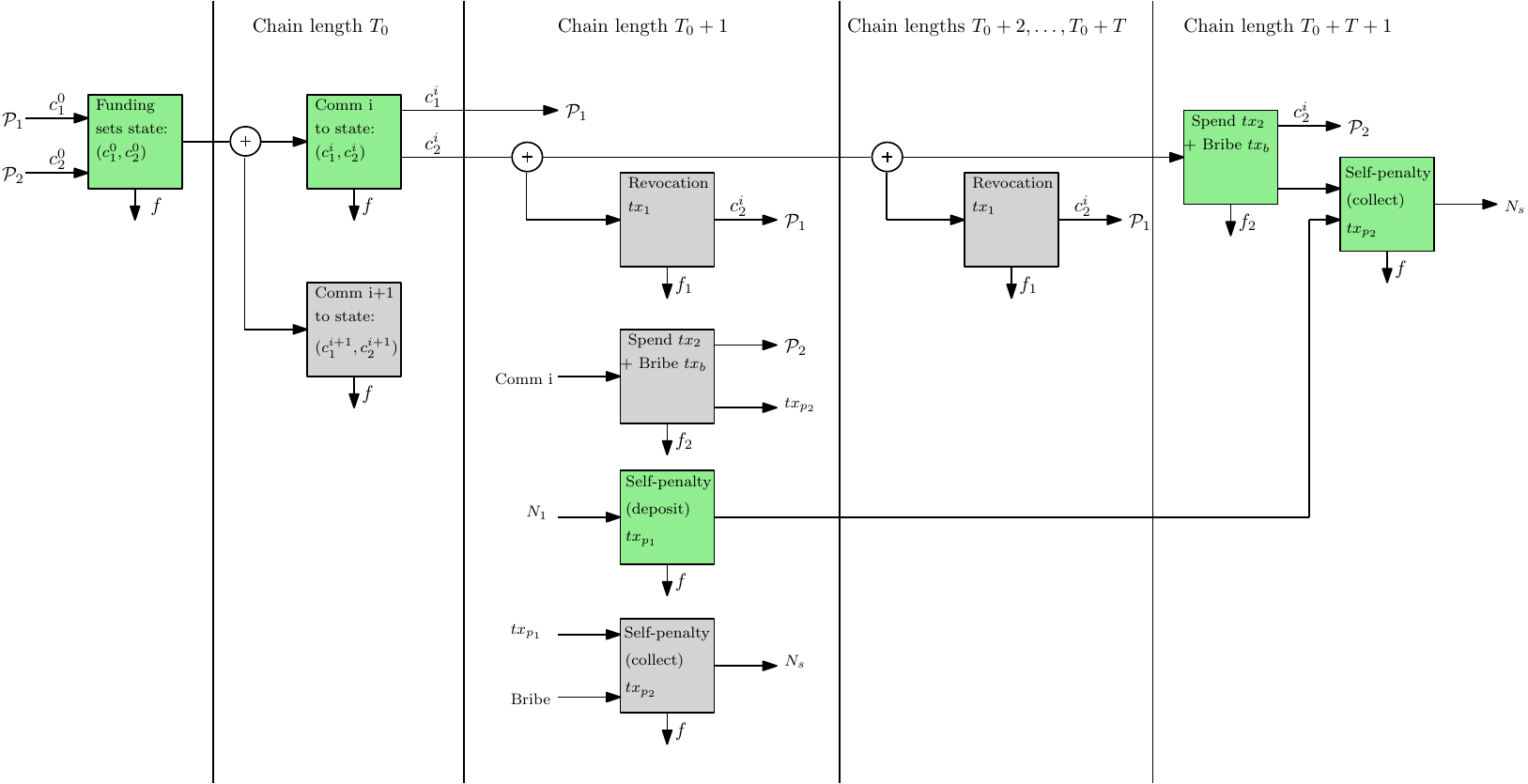}}
\caption{The $\attack$ attack. The green boxes indicate the transactions that should be put on-chain to run a successful $\attack$ attack. The grey boxes indicate the transactions that should be published on the mempool before the chain reaches a specific length. For instance, Spend transaction $tx_2$ has to be published on the mempool before the chain reaches length $T_0 + 1$, even though it can not be published on the blockchain until the chain reaches length $T_0 + T$. The arrows going into the boxes indicate the spending conditions of the transactions and the arrows going out of the boxes indicate how the funds of the boxes can be spent.}
\label{fig:attack:details}
\end{figure}



\section{Our Model}
\label{sec:our:model}
In this section, we gradually define our game that models the process of mining that takes into account the forks. A summary of our notation can be found in 
Table~\ref{tab:notation}.

\subsection{Preliminaries}
\label{sec:preliminary:definitions}
Let us begin by recalling the conditionally-timelocked output definition from~\cite{zeta}.
\begin{definition}[Conditionally timelocked output~\cite{zeta}]
A conditionally timelocked transaction output $txo(T_0,T,cond_1,cond_2)$ is a transaction output of a transaction $tx$ with spending condition $cond_1 \lor cond_2$. Condition $cond_1$ is not encumbered with any timelock and condition $cond_2$ is encumbered with a timelock that expires $T$ blocks after the block with height $T_0$, where $tx$ was published.
\end{definition}
The game with forks is defined for a fixed set of players (miners)  $\players = \left\{\player_1, \ldots, \player_n\right\}$ with a tuple of mining powers $\boldsymbol{\lambda} = (\lambda_1, \ldots, \lambda_n)$ and will last $R$ rounds.
Notice that we focus on proof-of-work blockchains that employ the so-called Nakamoto consensus, such as Bitcoin~\cite{nakamoto2008bitcoin}. We assume that in such environments miners (players) tend to form mining pools to (a) bypass the task of verifying transactions -- the pool's manager dispatches a list of valid transactions for inclusion -- and more importantly, (b) to guarantee a more stable income as individual mining carries substantial deviation. 
Empirical evidence supporting this assumption, drawn from Bitcoin, is detailed in Section~\ref{sec:data}.

For the rest of this work, each miner is assumed to either mine independently or stick to a selected mining pool throughout the execution of the game, i.e., we treat the mining pools as single players in the game.
As already mentioned, we assume that the game lasts for a fixed number of rounds. Alternatively, we could consider a scenario where the game lasts until the main chain reaches a predefined length. The first assumption is more suitable for the time periods when the block rate is constant. On the other hand, the second modeling approach is better for longer periods where the mining difficulty of blockchains is adjusted to achieve a given number of blocks within a given time unit. 

\subsubsection{Global State Object} We introduce a global state object $\state = \{ S_1, \ldots, S_{|S|} \}$  that describes a set of currently mined chains on the blockchain. 
Each $S_i$ consists of a \emph{list} (chain) of pairs $S_i = [\ (\block_1,\winner_1)$, $\ldots, (\block_{|S_i|},\winner_{|S_i|})]\ $ describing successfully mined blocks. In each pair $(\block_j,\winner_j) \in S_i$, $\block_j$ describes a set of transactions included in the block, and $\winner_j \in \players$ indicates a player that successfully mined the block.
\subsubsection{Allowed Actions}
We define the classes of possible actions in our game:
\begin{itemize}
    \item All chains in a state can be \emph{continued}. When the operation \emph{continue} is successful, a new pair $(\block,W)$ is appended to the continued chain in the global state object.
    \item Chains of length at least $1$ in the state can be \emph{forked}. Whenever one of the players successfully \emph{forks}, the new (duplicate) chain is created in the global state object in the following manner:
\begin{itemize}
\item the source fork is duplicated; and
\item a new block replaces the latest block in the duplicate. 
\end{itemize}
For instance, let $\state = \{\ [(\block_1, \winner_1), (\block_2, \winner_2)]\ \}$ be a current state with a single chain $S_1$. Then, after a successful fork of $S_1$ with $(B_3,W_3)$, one gets $\state = \{\ [(\block_1, \winner_1)$, $(\block_2, \winner_2)],\allowbreak [(\block_1, \winner_1), (\block_3, \winner_3)]\ \}$.
\end{itemize}

Notice that on existing blockchains, miners can fork a chain or mine on top of an arbitrary block in one of the existing chains. However, forking that starts at old blocks is less likely to outrun the longest chain. For that reason, we exclude this possibility from the game (following \cite{forking_discussion,temporary,forkanalysis}). In other words, miners in our model can fork only the last transaction on one of the chains, and then either the original chain or the fork becomes stable whenever it reaches a length equal to the length of the original chain plus one. The forks can be modeled differently, e.g., assuming a longer fork length or using a finite automaton definition. We expect our results to hold in the alternative modeling as well, but with different parameters of our solution would change.

\subsubsection{The Abandon Rule}
Let us define the abandon rule $\mathsf{abandon}: \states \rightarrow \states$ that is later used to 
abandon old chains no longer useful in the game. As we allow forking only the newest block in a chain, our  $\mathsf{abandon}$ rule will make each chain that outruns the length of other chains the only chain in the game. That is, for any $S_i \in \state$:
$\exists_{S_j \in \state}\  len(S_j) \geq len(S_i) + 1$  the abandon rule removes $S_i$ from the state $\state$.

\subsubsection{Types of Transaction Sets}
\label{sec:transaction:sets}
 Each block mined in the game (denoted as $\block$) includes only one of the transaction sets listed below:
\begin{itemize}
     \item      
     An unlimited amount of unrelated transaction sets $txs_u$ that contain average transactions $tx_u$ unrelated to any special transactions listed below. These transaction sets can be mined at any point in the game.
    \item A transaction set $txs_1$ that contains the transaction $tx_1$ that spends money of $txo$ under $cond_1$. As long as $txo$ is not spent, this transaction set can be mined on a chain of any length. The rest of the transactions in this transaction set are unrelated transactions $tx_u$.
    \item A (bribing) transaction set $txs_2$ that contains the transaction $tx_2$ that spends money of $txo$ under $cond_2$ and a bribing transaction $tx_b$. As long as $txo$ is not spent, this transaction set can be mined on a chain of length $\geq T$. The rest of the transactions in this transaction set are unrelated transactions $tx_u$.
    \item A special transaction set $txs_{p_1}$. In the first round of the game, one of the players (say $\player_1$) might decide to create a transaction set $txs_{p_1}$ with a self-punishment transaction $tx_{p_1}$ (see the description of the penalty mechanism in the Section~\ref{sec:details}). The player chooses the amount $P$, which he deposits to the transaction. The rest of the transactions in this transaction set are unrelated transactions $tx_u$.
    \item A special transaction set $txs_{p_2}$ with transaction $tx_{p_2}$. The transaction $tx_{p_1}$ assures that the player $\player_1$ that created the transaction $tx_{p_1}$ may collect back the deposit $P$ by publishing the transaction $tx_{p_2}$, but only after the transaction set $txs_2$ is published on the blockchain (see the Section~\ref{sec:details}). The rest of the transactions in this transaction set are unrelated transactions $tx_u$.
\end{itemize}

\subsubsection{Rewards}
\label{sec:rewards}
We assume that a miner, after successfully mining a transaction set $txs_i$ on the main chain, gets a reward $\reward(txs_i)$ equal to $\base + f_i+P$, where $\base$ is a constant block reward and $f_i$ is a sum of user fees input by users posting transactions in the transaction set $txs_i$. 
Whenever $txs_i$ contains a transaction that locks C coins from the miner's account, we set P to be equal to $-C$. Analogously, when a miner collects $C$ as one of the transactions from $txs_i$, we set the parameter $P$ to $C$.  
The reward for mining a block depends on the number of transactions within the block and their fees. The fee for a more complex transaction is typically higher as it occupies more space in a block. In this respect, we make the following assumptions that correspond to the current Bitcoin values (see Section~\ref{sec:fee} for more details):
\begin{itemize}
    \item Each unrelated transaction set $txs_u$ has on average $m$ transactions, its reward $$\reward(txs_u) = \base + f = \base + m \cdot \overline{f},$$ where $\overline{f}$ is an average user transaction fee. We also assume that $\overline{f} < 10^{-4}\base$.
    \item For other transaction sets with an uncommon functionality, e.g., $txs_{j}$, we assume it contains in total $m-c_j$ unrelated (average) transactions, a special transaction $tx_j$ that takes space of $c_j$ average transactions, where $c_j < m$. In total, $$\reward(txs_j) = \base + (m-c_j)\overline{f}+c_j\overline{f_j} + P=\base +f_j + P,$$ where $f_j = (m-c_j)\overline{f}+c_j\overline{f_j}$.
    The interpretation of the parameter $c_j$ is that it describes the number of transactions needed in a block to implement the uncommon functionality, each of them with fee $\overline{f_j}$.
\end{itemize}

Notice that in Section~\ref{sec:fee}, based on empirical data, we show how block rewards fluctuate in the real world. However, following~\cite{zeta}, we assume that standard transactions have a constant (average) reward and that all blocks have a constant number of transactions. In the list above, we refer to each standard transaction as an average transaction.

\subsubsection{Mining Power Distribution}

We assume the following mining power allocation $\lambda=(\lambda_1, \ldots, \lambda_n)$ among the players (see discussion in the Section~\ref{sec:experimental}):
\begin{itemize}
    \item There exists a single "strong" player (say player $s$) with mining power $\lambda_s \geq 20\%$. All other players have mining power smaller than $\lambda_s$.
    \item There exists a "relatively" strong player (say player $i$) with mining power $1\% < \lambda_i < 2\%$.
    \item We assume that all players with mining power less than $1\%$ have collective power at most $5\%$.
    \item The smallest mining power is of any miner in the game is $\lambda_{min} > 10^ {-100}$.
\end{itemize}

\subsubsection{Players' reluctance to believe a threat}
\label{sec:threat}
In the mining process, players can threaten other players that they will fork their blocks, once these blocks appear on the blockchain, as in the feather forking attacks. However, without any additional assumptions, there exist multiple solutions for such a setting~\cite{karakostas2024blockchain}.
To derive a single solution in our game, we make a conservative assumption that the players do not conduct the forking action if it \emph{can} result in financial losses to them. 
In other words, we accept only threats from a player who strictly profits from forking a selected transaction, i.e., the forking action is a dominating one for the player in this particular state.
\subsubsection{No Shallow Forks Conjecture} 
\label{sec:shallow:forks}
The Conjecture~\ref{conj:explicit:reward} below is a second assumption (together with the assumption that players are reluctant to believe a threat) that allows us to achieve a unique solution in the game with forks. In the conjecture, we assume that players have the option to fork a transaction only when they see an explicit opportunity of mining any other transaction with a \emph{higher} miner's fee\footnote{We denote that, alternatively to  Conjecture~\ref{conj:explicit:reward}, one could assume that the size of the mining fees in the game is limited, as excluding transactions with outstandingly high fees can also discourage the players from forking these transactions.}, initially blocked by the currently forked transaction. That is, we forbid shallow forks in the model.
\begin{conj}[No shallow forks]
\label{conj:explicit:reward}
At any point in the game $\game$, the players will not start a fork of a chain ending with a transaction set $txs_a$, unless they see an explicit opportunity to mine a fork containing at some point $txs_b$, initially blocked by $txs_a$ (e.g., by the conditionally timelocked output transaction mechanism, or the self-penalty mechanism). What is more, the players must be aware that $\reward(txs_b) > \reward(txs_a)$.
\end{conj}

Note that given the feasible transaction sets in the Section~\ref{sec:transaction:sets} and the reward structure defined in the Section~\ref{sec:rewards}, whenever $\reward(txs_2) > \reward(txs_1)$ and $\reward(txs_1) > \reward(txs_{p_1})$, according to the Conjecture~\ref{conj:explicit:reward}, the players in our game can attempt to fork only $txs_1$ to get $txs_2$ or $txs_{p_1}$ to mine $txs_1$. In other words, whenever $txs_u, txs_2,$ or $txs_{p_2}$ appear on one of the chains, they will not be forked.

\subsection{The Game}
Finally, we describe a game that models the process of PoW blockchain mining taking into account the option to conduct a forking of a block. The game proceeds in rounds; in each round, the miners can choose whether they want to continue mining one of the chains or they want to fork one of the chains. 
\begin{definition}[Conditionally timelocked game with forks]
A \emph{conditionally timelocked game with forks} $\game(\players,R)$ is a game with a finite set of players (miners) $\players = \left\{\player_1, \ldots, \player_n\right\}$, where $n = |\players|$, that lasts $R$ rounds. We define a tuple of mining powers $\boldsymbol{\lambda} = (\lambda_1, \ldots, \lambda_n)$ associated with the players, such that $\sum_{\lambda_i \in \lambda} = 1$. In the following, we will write $\game$, instead of $\game(\players,R)$, when $\players,R$ are obvious from the context.

Given the global state object, the set of possible actions, the rewards structure and the mining power distribution defined above, the game is played as follows: 
\begin{enumerate}
    \item The game starts with the state $\state = \{ [\ ]\}$ which is updated exactly $R$ times. All players are aware that this state is built upon a blockchain of height $T_0$ which includes an unspent conditionally timelocked transaction output $txo(T_0,T,cond_1,cond_2)$, where $T < R$.
    \item At each round $1 \leq r \leq R$, players $\player_i \in \players$ choose which of the subchains $\state_k \in \state$  they build upon, whether they will \emph{continue} or \emph{fork} this chain and which of the feasible blocks (built upon one of the transaction sets) they want to add in case they are declared as the winner. 
    Let $\actions\left(S,r\right)$ denote the set of all feasible actions for the state $(\state,r)$ described as triplets $(S, \decision, \mathsf{transaction\_set})$, where $S \in \state$, and $\decision \in \{\mathit{continue}, \mathit{fork}\}$. Based on $\boldsymbol{\lambda} = (\lambda_1, \ldots, \lambda_n)$, one player is declared as the winner in the round $r$, and the state object is modified accordingly.
    \item After each round, the abandon rule $\mathsf{abandon}$ is run on the current state.
\end{enumerate}

When the final round $R$ of the game is over, it finishes in some state $\state$, and rewards are given to the players.
By $\state^{*}$, let us denote the longest chain in the state $\state$. Whenever the state has multiple longest chains, $\state^{*}$ denotes the oldest of the longest chains of $\state$.
After the final round $R$, in the state $\state$, the reward given to a player $\player_i$ is:
$\reward_i(\state) = \sum_{(\block, W) \in \state^{*}: W = i} \reward(\block).$
\end{definition}
\subsubsection{Strategies}
Notice that given the set of players $\players$, the actions continue and mine defined above, and the set of transaction sets possible to mine, one can determine the set $\states$ of all states that may happen in the game.

A strategy $\sigma_i$ for a player $\player_i$ is given by a function mapping each pair  $(\state,r) \in \states \times [R]$ to a triplet feasible for this pair $(S, \decision, \mathsf{transaction\_set})$.\\
Let $\sigma$ denote a strategy profile of all players - a tuple of strategies of all players - i.e., $\sigma = (\sigma_1, \ldots, \sigma_n)$. Given a fixed index $i$, with $\sigma_{-1}$, we will denote a strategy profile of all players, but the selected $\player_i$.

The distribution of mining power among the players $\boldsymbol{\lambda} = (\lambda_1, \ldots, \lambda_n)$, the strategy profile $\sigma$, current state $\state \in \states$ and current round $r \in [R]$ define a probability distribution function $\boldsymbol{p}_{\boldsymbol{\lambda},\sigma, \state, r, r'}:~\states~\rightarrow [0,1]$ that assigns a probability that a certain state $\state' \in \states$ is activated after round $r' \in [R]$, where $r' > r$.
Given $\boldsymbol{p}_{\boldsymbol{\lambda},\sigma, \state, r, r'}$ and the $\reward$ function, we can define the utility (the expected reward) $\expectedutil_i(\sigma)$ of each player $i$, when strategy $\sigma$ is played.
We say $\sigma^* = (\sigma_1^*, \ldots, \sigma_n^*)$ is a Nash Equilibrium if for all players $\player_i \in \players$ it holds that $$\expectedutil_i(\sigma^*_i, \sigma_{-i}^*) \geq \expectedutil_i(\sigma_i, \sigma_{-i}^*),$$
for all alternative strategies $\sigma_i$ for the player $i$.

We denote by $\game^{\state, r}$ the subgame of the game $\game$ in a round $r$, at a state $\state$. We denote by $\expectedutil_i(\sigma, \state, r)$ the utility of a player $\player_i$ in this subgame, which is the expected reward for this player once the game is over.


\section{Analysis of \attack}\label{sec:analysis}
In this section, we formally analyze \attack where a bribe transaction $tx_2$ is published, large enough to bribe a chosen miner with the highest mining power, yet significantly smaller than the value required to directly persuade all miners to skip mining the transaction $txs_1$. The selected miner is then asked to threaten others with a fork if they add the unwanted transaction $txs_1$ to the blockchain. To make these threats credible, we implemented the self-penalty mechanism (see Section~\ref{sec:self:penalty}).

\subsection{About the proofs} In the proofs, we aim to find a \emph{dominating strategy} for a player $\player_i$ in a given state $\state$ and a round $r$, i.e., a strategy that outweighs other strategies of a selected player in the given state and round. As we will move from the very last round of the game till the first round of the game, we will be able to conclude our reasoning with a single NE of the full game.  Whenever needed, we use the mathematical induction technique to show that some choice of strategy is optimal for a sequence of rounds. Usually, the base case is the last round of the game and the induction step proves that if a given strategy is dominating in a round $k+1$, then it is also a dominating strategy in a round $k$.

\noindent When we compare how the player~$\player_i$ benefits from taking two distinct actions $A,B$ in a given state $\state$ and a round $r$, we often say that there exists a constant $C$ common for these strategies. 
To this end, we assume that action $A$ refers to some strategy $\sigma_a$ of the player $\player_i$, and action $B$ refers to some strategy $\sigma_b$ of the player~$\player_i$, such that $\sigma_a$ differs only in its definition from $\sigma_b$ on the selected state $\state$ and the selected round $r$. 
The utility of the player~$\player_i$ is the same for both strategies whenever in the state $\state$ and $r$ someone else than $\player_i$ is selected as the winner of the round. 
With $C$, we denote the utility of player~$\player_i$ multiplied by the probability of this event when player $\player_i$ is not the winner of the round.
This reasoning gives us an easy-to-use method to compare utility between the strategies $\sigma_a, \sigma_b$.
We can thus compare the utilities of the player $\player_i$ in the state $\state$ and round $r$ when the two distinct strategies $\sigma_a, \sigma_b$ are selected as:
\begin{flalign*}
    \expectedutil_i(\sigma_a, \state, r) = \lambda_i(\text{utility of the player 
 } \player_i \text{ when action A was taken})+C, \\
    \expectedutil_i(\sigma_b, \state, r) = \lambda_i(\text{utility of the player 
 } \player_i \text{ when action B was taken})+C.
\end{flalign*}

\subsection{Transaction Order in a Single Chain}
\label{sec:settled}
Although, due to the definition of the player's utility function, the bribing transaction ($txs_2$) may encourage the players not to skip mining transactions with high rewards during the dispute period (e.g. $txs_1$), we first show that once timelock is over and $txs_1$ was not mined, the players will follow the default strategy to mine transactions with highest rewards first.
\begin{restatable}[]{lemma}{orderedtransactions}
Let $\game(\state, T+1)$ be a subgame in a state $\state = \{S\}$,
where the state $\state$ contains a single chain $S$ and the transaction set $txs_{p_1}$ was mined in the first round. In the next $T-1$ rounds, miners mined unrelated transactions sets $txs_u$. Furthermore, it holds that $$\reward(txs_2) > \reward(txs_1)>\reward(txs_u)\textit{ and }\reward(txs_{p_2})>\reward(txs_u).$$  Then the dominating strategy for all players in the subgame $\game(\state,T+1)$ is to mine transaction sets in the following  order: $txs_2, txs_{p_2}$, and for the rest of rounds $txs_u$.
\label{claim:transactions:order}
\end{restatable}
\begin{proof}
As $txs_{p_1}$ was mined in the first round, then in any round after the round $T+1$ there are available to mine the following transaction sets:
\begin{itemize}
    \item mutually exclusive $txs_1$ and $txs_2$ with expired timelock,
    \item  one $txs_{p_2}$ that can be mined only after $txs_2$ appears on blockchain,
    \item and an unlimited amount of unrelated transaction sets $txs_u$. 
\end{itemize} 
Since only $txs_1$ and $txs_2$ are mutually exclusive and   $\reward(txs_1) < \reward(txs_2)$, then by Conjecture~\ref{conj:explicit:reward}, whenever $txs_2$, $txs_{p_2}$ or $txs_u$ appear on the blockchain, they will not be forked. Thus only $txs_1$ may be forked in the subgame.

Once both $txs_{2}, txs_{p_2}$ are on the chain, miners can not mine any special transaction sets, and all of the miners will mine $txs_u$ till the end of the game.

Next we show that for any round $r \in \{T+2, \ldots, R\}$, in a state $\state'$ created by extending the chain in $\state$ with $txs_u$ and one $txs_2$ at any point in this chain,
 the dominating strategy for all players is to mine $txs_{p_2}$ first if it was not mined until this point. We will prove it by induction. The statement trivially holds in the last round $R$, because $txs_{p_2}>txs_u$. Now, assuming that it holds in round $R - k$, we prove that it also holds in round $R - k-1$. Any player $\player_i$ will be chosen with probability $\lambda_i$ as the winner of the round. The utility of the player $\player_i$ following a strategy $\sigma_{p_2}$ that first mines $txs_{p_2}$ in the state $\state'$ is:
    \begin{flalign*}
        \expectedutil_i(\sigma_{p_2}, \state', R-k-1) = \lambda_i(\reward(\state') + f_{p_2} + \base + \lambda_i k(f+\base))+C,
    \end{flalign*}
    for some constant $C$ that describes the expected reward of $\player_i$ in case he/she is not chosen as the winner of this round.

\noindent The utility of the player $\player_i$ following a strategy $\sigma_{u}$ that first mines $txs_u$ in the state $\state'$ is:
    \begin{equation*}
\expectedutil_i(\sigma_{u}, \state', R-k-1) = \begin{cases}
\lambda_i(\reward(\state') + f_{u} + \base)+C &\text{when $k=0$}\\
\lambda_i(\reward(\state') + f_{u} + \base + \lambda_i(f_{p_2} + \base) + \lambda_i (k-1)(f+\base))+C &\text{when $k\geq1$}
\end{cases}
\end{equation*}
From the above it follows that 
\begin{flalign*}
\expectedutil_i(\sigma_{p_2}, \state', R-k-1) > \expectedutil_i(\sigma_{u}, \state', R-k-1).
\end{flalign*}

Next we show that for any $r \in \{R, \ldots, T+1\}$, in a state $\state''$ created by extending the chain in $\state$ with $txs_u$, the dominating strategy for all players is to mine $txs_2$ first if it was not mined until this point.
Observe that in the state $\state''$ the miners can only mine $txs_u, txs_1$ or $txs_2$.
The statement trivially holds in the last round. Now, assuming that it holds in round $R - k$, we prove that it also holds in round $R - k+1$. Any player $\player_i$ will be chosen with probability $\lambda_i$ as the winner of the round, and with probability $1-\lambda_i$ someone else will be selected as the winner of the round.
Then for some constant $C$, the utility of the player $\player_i$ in a strategy $\sigma_{2}$ that first mines $txs_{2}$ is
 \begin{equation*}
\expectedutil_i(\sigma_{2}, \state', R-k+1) = \begin{cases}
\lambda_i(\reward_i(\state'') + f_{2} + \base)+C &\text{when $k=0$}\\
\lambda_i(\reward_i(\state'') + f_{2} + \base + \lambda_i (f_{p_2}+\base))+C &\text{when $k=1$} \\
\lambda_i(\reward_i(\state'') + f_{2} + \base + \lambda_i (f_{p_2}+\base) + (k-2)\lambda_i(f+\base))+C &\text{when $k\geq2$}
\end{cases}
\end{equation*}

Whereas the utility of the player $\player_i$ in a strategy $\sigma_{1}$ that first mines $txs_{1}$ is

 \begin{equation*}
\expectedutil_i(\sigma_1, \state'', R-K+1) \leq  \begin{cases}
\lambda_i(\reward_i(\state'') + f_{1} + \base)+C &\text{when $k=0$}\\
\max\{ \lambda_i (\reward_i(\state'') + f_1+\base + \lambda_i(f+\base)) + C, &\text{when $k=1$} \\
\textit{      }\lambda_i (\reward_i(\state'') + f_2+\base) + C\} & \\
\max\{ \lambda_i (\reward_i(\state'') + f_1+\base + 2\lambda_i(f+\base)) + C, &\text{when $k=2$} \\
\textit{      }\lambda_i (\reward_i(\state'') + f_2+\base+\lambda_i(f_{p_2}+\base)) + C\} & \\

\max\{ \lambda_i (\reward_i(\state'') + f_1+\base + k\lambda_i(f+\base)) + C, &\text{when $k\geq3$} \\
\textit{      }\lambda_i (\reward_i(\state'') + f_2+\base+\lambda_i(f_{p_2}+\base)) +(k-3)\lambda_i(f+\base)) C\} & \\
\end{cases}
\end{equation*}
It is again easy to see that $\expectedutil_i(\sigma_{2}, \state'', R-k+1) > \expectedutil_i(\sigma_1, \state'', R-K+1)$.
\end{proof}
The details of the proof of the above Lemma imply the following result.
\begin{lemma}
Let $\game(\state, T+1)$ be a subgame in a state $\state = \{S\}$, 
where the state contains a single chain $S$ where in $T$ rounds miners mined unrelated transaction sets $txs_u$. Furthermore, for all miners $$\reward(txs_2) > \reward(txs_1)>\reward(txs_u).$$ Then the dominating strategy for all players in the subgame $\game(\state,T+1)$ will result in the following transactions order: $txs_2$, and for the rest of rounds $txs_u$.
\label{claim:transactions:order:two}
\end{lemma}
\subsection{Decisions of an Individual Miner are Consistent}
\label{sec:without:reward}
 In this section, we show that without a high-cost reward $f_2$, once someone is successful with mining $txs_1$, the miner will continue mining this chain, as it might be too costly for the miner to lose the block reward that he already mined. As $txs_1$, $txs_2$ is the only pair of conflicting transactions in the game whenever $txs_{p_1}$ was not created, it follows from Conjecture~\ref{conj:explicit:reward} that the forks may occur only when one miner successfully mines $txs_1$, and the other player wants to profit from mining $txs_2$. Thus, in the following, we study the behavior of the players whenever one of the players decides to mine $txs_1$.
\begin{restatable}[]{thm}{continuemining}
\label{thm:continue:mining}
Assuming subgame $\game(\state, r)$ in a state $\state$ with a single chain of length $r \leq T$, formed until round $r$ where player $\player_j$  mined $txs_1$ in the last round, and $txs_{p_1}$ is not on the chain, then the player $\player_j$ will $\continue$ to mine this chain unless $f_2 - f \geq f_1+\base$, even when other miners decide to fork the chain with $txs_1$ and $\continue$ mining the new subchain created during the fork.
\end{restatable}
\begin{proof}

At every point of the game,
each player $\player_i$ can choose a strategy for the remaining $M$ rounds to collect at least $ M \lambda_i (f+\base)$ 
if he simply always chooses to mine $txs_u$ from this point.
\noindent Thus, whenever $txs_1$ was just mined by $\player_j$ and $R-r$ rounds are left till the end of the game, then for some $C$:
\begin{itemize}
\item $\player_j$ chooses a strategy $\sigma_1$ where he continues the current chain of the state $\state$, thus: 
    \begin{flalign*}
    \expectedutil_j(\sigma_1, \state, r+1) \geq \lambda_j(\reward_j(\state) +f+\base+ (R-r-1) \lambda_j (f+\base)) + C.
    \end{flalign*}
    \item when the $\player_j$ "forks" himself, then at least one of the blocks $txs_1$ or $txs_u$ will be canceled out in the final chain, therefore for any strategy $\sigma_2$ that involves forking $txs_1$:
    \begin{flalign*}
    \expectedutil_j(\sigma_2, \state, r+1) \leq \lambda_j(\reward_j(\state) - (f_1+\base)+f_2+\base+ (R-r-1) \lambda_j (f+\base)) + C.
\end{flalign*}
\end{itemize}
In conclusion 
$\expectedutil_j(\sigma_1, \state, r+1) >\expectedutil_j(\sigma_2, \state, r+1)$, 
unless $f_2 - f \geq f_1+\base$.

Now, since $\player_j$  that already mined $txs_1$ will not fork himself in the first round, it is easy to see that the same follows in the next round.

\end{proof}

\subsection{Only a High-Cost Reward May Encourage Miners to Fork}
\label{sec:high:reward}
The next result shows that it is not possible to credibly threaten with forks without high fees. In particular, we show that for any miner with mining power $\lambda_j$ that \emph{considers} mining the block $txs_1$, any forking threat in the game where $txs_{p_1}$ was not created, will not be credible unless $f_2-f \geq \lambda_j(f+\base)$, as the miner that mined the transaction will continue to mine his transaction.

\begin{restatable}[]{thm}{noforks}
\label{thm:no:forks}
Let $\game(\state, r+1)$ be a subgame in a state $\state$ that contains only a single chain of length $r$ consisting of $r-1$ unrelated transaction sets $txs_u$ and one (just mined) $txs_1$ (mutually exclusive with $txs_2$ with $\reward(txs_2) > \reward(txs_1) > \reward(txs_u)$) mined by some miner $\player_j$. The $txs_{p_1}$ was not created and $r \leq T$. Other miners will not fork $txs_1$, unless $f_2-f \geq \lambda_j(f+\base)$, where $\lambda_j$ is the mining power of the miner $\player_j$.
\end{restatable}
\begin{proof}
    First, we introduce a notation of the states which are relevant for the analysis of the above subgame.
    We say that a game is in the state $\state_{0,-1}$ 
    when $txs_1$ was just mined, and no-one forked the chain with $txs_1$ yet.
    Let's observe, that in this notation the state $\state_{0,-1}$ is equal to the initial state $\state$.
    We say that a game is in a state $\state_{k,l}$ ($k\geq0, l\geq0$), whenever the chain defined in $\state$ reaches length $k$ compared to the initial state and the chain from $\state$ is forked and the fork reaches length $l$ compared to the initial state.
    
   As $txs_1$ may be forked only if there does not exist any other chain where it can be mined (Conjecture~\ref{conj:explicit:reward}),
   at most $2$ chains may be formed in the game. Whenever $k$ or $l$ reach length $1$, one of these chains will be closed. For any player $\player_i$ that we currently consider, we say that in the game is in the state $\state_{k,l}^{i, u,d}$, if the player already mined $u$ blocks in the chain with $txs_1$ and $d$ blocks in the fork chain ($k,l$ describe length of the chains as above).
   

First, we analyze the behaviour of some other player $\player_i \neq \player_j$. $\player_i$ will be chosen with probability $\lambda_i$ as the winner of the round, and with probability $1-\lambda_i$ someone else will be selected as the winner of the round. Then for some constant $C$, the utility of the player $i$ in all strategies $\sigma_1$ that continue mining $txs_1$ in $\state$ can be described as:
\begin{flalign*}
    \expectedutil_i(\sigma_1, \state, r) = \lambda_i \expectedutil_i(\sigma_1, \state_{1,-1}^{i,1,0}, r+1) + C.
\end{flalign*}
 And the utility of the player $i$ in all strategies $\sigma_2$ that fork $txs_1$ in $\state$ can be described as:
 \begin{flalign*}\expectedutil_i(\sigma_2, \state, r) = 
 u_{i,0,-1}^{0,0}[\sigma_1]=\lambda_i \expectedutil_i(\sigma_2, \state_{0,0}^{i,0,1}, r+1) + C.\end{flalign*}

 By Theorem~\ref{thm:continue:mining}, the player $\player_j$ that mined $txs_1$, will continue to mine the chain with $txs_1$ in the next two rounds. Furthermore, as already mentioned, any miner $\player_i$ that decides to continue the chain with $txs_1$, can profit at least $\lambda_i(f+\base)$ in each round till the end of the game, by simply mining $txs_u$. From the above it follows that.
    \begin{flalign*}
        \expectedutil_i(\sigma_1, \state_{1,-1}^{i,1,0}, r+1) \geq \reward_i(\state)+ f + \base + (R-r)\lambda_i(f+\base),
    \end{flalign*}
    
 \begin{flalign*}
     \expectedutil_i(\sigma_2, \state_{0,0}^{i,0,1}, r+1) \leq \reward_i(\state) + \max\{(f_2 + \base)(1-\lambda_j) + (M-1)\lambda_i(f+\base) \\ (f + \base)(1-\lambda_j) + \lambda_i(f_2+\base)+ (R-r-1)\lambda_i(f+\base) \}.
 \end{flalign*}
    In conclusion $\expectedutil_i(\sigma_1, \state_{1,-1}^{i,1,0}, r+1) > \expectedutil_i(\sigma_2, \state_{0,0}^{i,0,1}, r+1)$
    unless $f_2-f \geq \lambda_j(f+\base)$.
\end{proof}
\subsection{Without a High-Cost Reward, All Players Mine $txs_1$}
\label{sec:mine:txsone}
As we already observed, once $txs_1$  is mined, it will not be forked unless the bribing fee is sufficient. We will show that for a sufficiently large number of rounds $T$, all of the players will mine transaction set $txs_1$ in the first round. A similar result was introduced in~\cite{timelocked}, but we prove that this result still holds in the game with forks.

\begin{restatable}[]{thm}{txsone}
Let $\game(\state,1)$ be a subgame where none of the miners decides to create $txs_{p_1}$ before the first round, and the bribing fee is not too high, i.e. $f_2-f < 10^{-2}(f+\base)$. What is more $f_1 > f$, and if we define $Y = \sum_{j=i: \lambda_j > 0.01, f_2-f <\frac{f_1-f}{\lambda_{j}}}^{|\players|} \lambda_j$, then $T,Y$ are big enough, such that $(1- 1.01(1-Y)^T) > 0$. Every miner with $\lambda_j > 0.01$ will decide to mine $f_1$ in the first round.
\label{thm:txs1}
\end{restatable}
\begin{proof}
In the game where none of the miners decides to create the transaction set $txs_{p_1}$, miners may choose to mine $txs_u$ and $txs_1$ in all rounds, or $txs_2$ only after round $T$. 
Now, since the game contains only one pair of mutually exclusive transactions $txs_1, txs_2$ with $\reward(txs_2) > \reward(txs_1)$, then by Conjecture~\ref{conj:explicit:reward} players can start to fork only when $txs_1$ appears on the blockchain.
What is more, since $f_2-f < 10^{-2}(f+\base)$, by Theorem~\ref{thm:no:forks}, whenever some player with $\lambda_j > 10^{-2}$ successfully mines $txs_1$ in a chain of length $\leq T$, none of the players will decide to fork his block.

We prove that in the above game miners with collective mining power at least $Y$ will decide to mine $txs_1$ in rounds $\{T, T-1, \ldots, 1\}$ if not mined up to this point.
Let's take any miner with $\lambda_i > 10^{-2}$ that makes a decision in round $T-k$, for $k \in \{0,\ldots, T-1\}$. As already mentioned, once he successfully mines the block $txs_1$, it will not be forked. In round $T$, whenever the block $txs_1$ was not mined, then the miners had only mined $txs_u$ so far ending up in a state $\state_T$. Then for some constant $C$, the utility of the player $\player_i$ in all strategies $\sigma_1$ that choose to mine $txs_1$ in $\state_T$, and all strategies that choose to mine $txs_u$ in $\state_T$:
\begin{flalign*}\expectedutil_i(\sigma_1, \state_T, T) \geq \lambda_i(\reward_i(\state_T)+f_1+\base+\lambda_i (R-T-1) (f+\base)) + C_{},
\end{flalign*}
\begin{flalign*}
    \expectedutil_i(\sigma_2, \state_T, T) \geq \lambda_i(\reward_i(\state_T)+f+\base+\lambda_i(f_2+\base) + (R-T-2)\lambda_i(f+\base)) + C.
\end{flalign*}
Now, $\expectedutil_i(\sigma_2, \state_T, T) < \expectedutil_i(\sigma_1, \state_T, T)$ only if $(*) f_2-f \geq \frac{f_1-f}{\lambda_i}$.  This implies that miners with collective mining power at least $Y$ will prefer to mine $txs_1$ in this round.

In round $T-k$, where $k>0$, whenever the block $txs_1$ was not mined, then the miners had only mined $txs_u$ so far, ending up in a state $\state_{T-k}$. Then for some constant $C$, the utility of the player $i$ in all strategies $\sigma_1$ that choose to mine $txs_1$ in $\state_T$, and all strategies that choose to mine $txs_u$ in $\state_T$:
\begin{flalign*}
    \expectedutil_i(\sigma_1, \state_{T-k}, T-k) =  \lambda_i(\reward_i(\state_{T-k})+f_1+\base+\lambda_i(k-1)(f+\base)+\\ \lambda_i (R-T+k-1)(f+\base)) + C,
    \end{flalign*}
    \begin{flalign*}
    \expectedutil_i(\sigma_2, \state_{T-k}, T-k) \leq  \lambda_i(\reward_i(\state_{T-k})+f+\base +\lambda_i(k-1)(f+\base)+ \lambda_i(f_2+\base)+\\(R-T+k-2)\lambda_i(f+\base)) + C.
\end{flalign*}
Similiary, the above equation implies that at least  $Y$ miners will prefer to mine $txs_1$ in this round.

Now, after the first round there are $T$ rounds till the moment of mining $txs_2$, player's $\player_i$ benefit from mining $txs_1$ (with $\lambda_i > 0.01$) and not waiting for $txs_2$ is at least $benefit = \lambda_i (f_1+\base) - \lambda_i (1-Y)^T(f_2+\base)$, and since $f_2-f < 10^{-2}(f+\base)$, then $benefit \geq \lambda_i(f_1+\base - (1-Y)^T)(1.01f+1.01B)$. Now, assuming that $f_1 > f$ we have $benefit \geq \lambda_i(f (1- 1.01(1-Y)^T)) + \base(1-1.01(1-Y)^T)$. This implies that whenever $(1- 1.01(1-Y)^T) > 0$, then all miners with $\lambda_i > 0.01$ will mine $txs_1$ in the first round.
\end{proof}

\noindent A similar results holds in any state where sufficiently large number of $T-r+1$ rounds are left till the round $T$, $txs_{p_1}$ was not created in the game (or $txs_u$ was mined in the first round), and $txs_1$ was not mined yet. We leave it as a lemma without a proof.
\begin{lemma}
Given a game with forks $\game(\state,r)$ with $r < T$, where none of the miners decides to create $txs_{p_1}$ before the first round (or $txs_u$ is mined in the first round), and the bribing fee is not too high, i.e. $f_2-f < 10^{-2}(f+\base)$ and  given $Y = \sum_{j=i: \lambda_j > 0.01, f_2-f <\frac{f_1-f}{\lambda_{j}}}^{|\players|} \lambda_j$; $T-r+1,Y$ are big enough, such that $(1- 1.01(1-Y)^{T-r+1}) > 0$, every miner with $\lambda_j > 0.01$ will decide to mine $f_1$ in this round.
\label{lemma:txs1}
\end{lemma}

\subsection{Discouraging Miners to Mine $txs_1$}
\label{sec:define:penalty}
In the previous sections, we have shown that it is rather expensive to force the players not to mine $txs_1$ in the first round, even when the players can fork this transaction. In this section, we leverage the self-penalty mechanism introduced in Section~\ref{sec:self:penalty}. The proof is inductive, and its base case starts in round $T$. For each round, we first show that the miner $\player_s$ with the highest mining power $\lambda_s$ will not mine $txs_1$, as we assume that $f_2-f > \frac{f_1-f}{\lambda_{s}}$. Next, given a sufficiently large penalty $P >\lambda_s(f+\base)$, we show that the selected player $\player_s$ will fork the transaction $txs_1$, once it appears on the blockchain, even though it poses a risk of losing the block reward. Finally, we show that in this round all players other than the player $\player_s$ are afraid to mine $txs_1$, when the self-penalty transaction is on the chain.
\begin{restatable}[]{thm}{txstwo}
Let $\game(\state,2)$ be a subgame where $txs_{p_1}$ defined by a  player $\player_s$ with mining power $\lambda_s$ was mined in the first round with $P >\lambda_s(f+\base)$. What is more $f_2-f > \frac{f_1-f}{\lambda_{s}}$, and $\frac{f+\base}{f_1+\base} > 1-\lambda_s^2$. None of the miners will decide to mine $txs_1$ in rounds $2, \ldots, T$.
\label{thm:txs2}
\end{restatable}
\begin{proof}
Let us denote with $\state_{p_1}^{T}$ a state containing a chain of length $T-1$, where $txs_{p_1}$ was created by $\player_s$ with mining power $\lambda_s$ and with the deposit parameter $P~\geq~f+\base- \lambda_s(f_2+\base)$. All other mined transaction sets in $\state_{p_1}^{T}$ are $txs_u$. With $\state_{p_1, f_1}^{T}$ we denote a state $\state_{p_1}^{T}$ with transaction set $txs_1$ added at the top of the chain in the state $\state_{p_1}^{T}$.

Let us denote with $\sigma_1$ all strategies where player $\player_s$ continues to mine the chain with $txs_1$ in $\state_{p_1, f_1}^{T}$, and let $\sigma_2$~denote all strategies, where he forks $txs_1$ in $\state_{p_1, f_1}^{T}$ and then continues to mine the newly created chain (we assume that the transaction $txs_1$ was not mined by the miner $\player_s$, as when $f_2-f > \frac{f_1-f}{\lambda_{s}}$, then as shown in the proof of the Theorem~\ref{thm:txs1}, the player $\player_s$ will choose to mine $txs_u$ in $\state_{p_1}^{T}$). It is easy to see that once the player $\player_s$ successfully forks $txs_1$ and creates a new chain, all miners other than the miner that mined $txs_1$ will prefer to mine the newly created chain. As the miner who mined the transaction $txs_1$ has the mining power less than $\lambda_s$, we can say that once $txs_1$ is forked, the miners with the mining power at least $1-\lambda_s$ will continue mining the new chain. In summary, for some constant $C$, given similar arguments as in the previous proofs:
\begin{flalign*}
    \expectedutil_s(\sigma_1, \state_{p_1, f_1}^{T}, T+1) \leq \lambda_s (\reward_s(\state_{p_1, f_1}^{T})+f+\base + (R-T-1) \lambda_s (f+\base) - P)+C,
\end{flalign*}
\begin{flalign*}\expectedutil_s(\sigma_2, \state_{p_1, f_1}^{T}, T+1) \geq \lambda_s(\reward_s(\state_{p_1, f_1}^{T})+(1-\lambda_s)(f+B)+\lambda_s(f_2+B)+\\ (R-T-2)\lambda_s (f+\base)) + C,\end{flalign*} as after two rounds, the player may at least mine $txs_u$ in each round.
This implies that for $P >\lambda_s(f+\base)$ the committed player prefers to choose strategy $\sigma_2$ and fork in $\state_{p_1, f_1}^{T}$.

\noindent Now, we consider the strategy for every player $\player_i$ such that $i \neq s$ and $\lambda_i < \lambda_s$ in state~$\state_{p_1}^{T}$.
Let us denote strategies where player $\player_i$ chooses to mine $txs_1$ in $\state_{p_1}^{T}$ with $\sigma_1$ and when he mines $txs_u$ in $\state_{p_1}^{T}$ with $\sigma_u$.
Then for some $C$ following inequalities are satisfied:
\begin{flalign}
    \nonumber
    \expectedutil_i(\sigma_u, \state_{p_1}^{T}, T) \geq \lambda_i [ \reward_i(\state_{p_1}^{T}) + (f + \base) + (R-T-2)\lambda_i*(f+\base) +\\ 
    \nonumber
    \lambda_i*(f_2+\base)] + C,
\end{flalign}
\begin{flalign*}
    \expectedutil_i(\sigma_1, \state_{p_1}^{T}, T) \leq \lambda_i[\reward_i(\state_{p_1}^{T}) + (f_1+\base)\prob^{survive}[\state_{p_1, f_1}^{T}] + (R-T-1)\lambda_i*(f+\base) ]]+ C.\end{flalign*}
Where $\prob^{survive}[\state]$ is the probability that the last block of the single chain in the state $\state$  remains a part of the chain in the last round $R$.  Since in $\state_{p_1, f_1}^{T}$ player $\player_s$ will prefer to fork the transaction set $txs_1$, then $\prob^{survive}[\state_{p_1, f_1}^{T}] \leq
 1-\lambda_s^2$. 
 
 \vspace{0.2cm} \noindent Now, whenever $1-\lambda_s^2< \dfrac{f+\base+\lambda_i(f_2+\base)+\lambda_i(f+\base)}{f_1+\base}$, then all players prefer to mine $txs_u$.
 
 \vspace{0.2cm} \noindent However~$\dfrac{f+\base+\lambda_i(f_2+\base)+\lambda_i(f+\base)}{f_1+\base} > \dfrac{f+\base}{f_1+\base}$, and $\dfrac{f+\base}{f_1+\base} > 1-\lambda_s^2$ by the assumption.
The above implies that on a chain of length $T$, every player mines $txs_u$.

\noindent For any chain of length $k = 3, \ldots, T-1$ we prove that if everyone mines only $txs_u$  on chain of length $k+1$ (i.e. strategy for every miner on $\state_{p_1}^{k+1}$ is to mine $txs_u$) or everyone forks on chain of length $k+1$ with $txs_1$ then the same assumptions hold for chains of length $k$.

\noindent First we prove that on a chain of length $k$ player $\player_s$ will fork $txs_1$, once it appears on the blockchain. We inductively assume that the assumptions hold for chains of length $k+1, \ldots, T$.

\noindent Again, once the player $\player_s$ successfully forks $txs_1$ and creates a new chain, all miners other than the miner that mined $txs_1$ will continue mining the new chain. Let us denote with $\sigma_1$ all strategies where player $\player_s$ continues to mine $txs_1$ in $\state_{p_1, f_1}^{k}$, and with $\sigma_2$ all strategies, where he forks $txs_1$ in $\state_{p_1, f_1}^{k}$. Then for some constant $C$:
\begin{flalign*}
    \expectedutil_s(\sigma_1, \state_{p_1, f_1}^{k}, k+1) \leq  \lambda_s (\reward_s(\state_{p_1, f_1}^{k})+f+\base + (T-k) \lambda_s (f+\base) - P)+C,
\end{flalign*} 
\begin{flalign*}\expectedutil_s(\sigma_2, \state_{p_1, f_1}^{k}, k+1) \geq \lambda_s(\reward_s(\state_{p_1, f_1}^{k})+(1-\lambda_s)(f+B)+\\ \lambda_s(f_2+B)+ (T-k-1)\lambda_s (f+\base)) + C.\end{flalign*} as after two rounds, the player may always mine $txs_u$ in each round. 
This implies that for $P >\lambda_s(f+\base)$ the committed player prefers to choose strategy $\sigma_2$ and fork in~$\state_{p_1, f_1}^{k}$.

\noindent Next, we analyse the optimal strategy for any player $\player_i$ such that $i \neq s$, in $\state_{p_1}^{k}$. Let us denote strategies where player $\player_i$ chooses to mine $txs_1$ in $\state_{p_1}^{k}$ with $\sigma_1$ and with $\sigma_u$ when he wants to mine $txs_u$ in $\state_{p_1}^{k}$.
It holds for some $C$, that:
\begin{flalign*}
    \expectedutil_i(\sigma_u, \state_{p_1}^{k}, k) \geq \lambda_i [\reward_i(\state_{p_1}^{k})+ (f + \base)\cdot 1 + (T-k-2)\lambda_i*(f+\base) + \lambda_i*(f_2+\base)] + C,
\end{flalign*}
\begin{flalign*}\expectedutil_i(\sigma_1, \state_{p_1}^{k}, k) \leq \lambda_i[\reward_i(\state_{p_1}^{k})+(f_1+\base)\prob^{survive}[\state_{p_1, f_1}^{k}] + (T-k-1)\lambda_i*(f+\base) ]]+ C.\end{flalign*}
As discussed previously, whenever $1-\lambda_s^2< \frac{f+\base+\lambda_i(f_2+\base)+\lambda_i(f+\base)}{f_1+\base}< \frac{f+\base}{f_1+\base}$ is satisfied, then all players prefer to mine $txs_u$.

\noindent Now, for the chain of length $2$, whenever $txs_1$ is mined it will be forked by a similar argument as discussed above, as once $txs_1$ appears on blockchain, $txs_{p_1}$ cannot be forked. 
The analysis of the following case is conducted to asses the behaviour of players in the second round of the game.

\vspace{0.2cm}\noindent\textbf{Case 1.} Player $\player_s$ mined $txs_{p_1}$ in the state $\state_{p_1}$.

\vspace{0.2cm}\noindent\textbf{Case 1.A.} When $txs_1$ was just mined by the player $\player_s$ and $R-1$ rounds are left till the end of the game, then for some $C$:
\begin{itemize}
\item by $\sigma_1$ we denote strategies where player $\player_s$ decides to continue the current chain of the state, adding $txs_1$:  
\begin{flalign*}
    \expectedutil_s(\sigma_1, \state_{p_1}, 2) \leq \max \{
    \lambda_s(\reward_s(\state_{p_1}) +f_1+\base+ (R-2) \lambda_s (f+\base)) + C,\nonumber\\
    \lambda_s(\reward_s(\state_{p_1}) + (R-3) \lambda_s (f+\base)+\lambda_s (f_2+\base)+P) + C
    \}\nonumber
    \end{flalign*}
    \item by $\sigma_2$ we denote strategies where player $\player_s$ decides to continue the current chain of the state, adding $txs_u$:  \begin{flalign*}
    \expectedutil_s(\sigma_2, \state_{p_1}, 2) \geq \lambda_s(\reward_s(\state_{p_1}) +f+\base+ (R-3) \lambda_s (f+\base) + \lambda_s (f_2+\base)+P) + C\nonumber
\end{flalign*}
    \item by $\sigma_{f_1}$ - we denote a strategy that chooses to fork $txs_{p_1}$ and mine $txs_1$ - 
    \begin{flalign*}
    \expectedutil_s(\sigma_{f_1}, \state_{p_1}, 2) \leq \max\{ \lambda_s(\reward_s(\state_{p_1}) + (f_1+\base)+(R-2) \lambda_s (f+\base)-\base) + C,\nonumber\\ 
    \lambda_s(\reward_s(\state_{p_1}) + (R-2) \lambda_s (f+\base)+\lambda_s (f_2+\base)+P-\base) + C \nonumber 
    \} 
\end{flalign*}
    \item by $\sigma_{f_u}$ - we denote a strategy that chooses to fork $txs_{p_1}$ and mine $txs_u$
    \begin{flalign*}
    \expectedutil_s(\sigma_{f_u}, \state_{p_1}, 2) \leq \max\{\lambda_s(\reward_s(\state_{p_1}) + f+\base+ (R-3)\lambda_s(f+\base)+\lambda_s(f_2+\base)-\base),\nonumber \\ 
    \lambda_s(\reward_s(\state_{p_1}) +\lambda_s(f_1+\base)+ (R-3)\lambda_s(f+\base)-\base) \nonumber
    \}
    \end{flalign*}
\end{itemize}
In summary, the strategy $\sigma_2$ is dominating and player $\player_s$ who mined $txs_{p_1}$, will prefer to continue mining $txs_{p_1}$. Note that the above analysis will be similar, even if some player other than the $\player_s$ successfully starts a fork of $txs_{p_1}$. 

\vspace{0.2cm}\noindent\textbf{Case 1.B.} Further, for any player $\player_i$, where $i \neq s$, for some constant $C$:

\begin{itemize}
\item by $\sigma_1$ we denote strategies where $\player_i$ decides to continue the current chain of the state, adding $txs_1$:  \begin{flalign*}
    \expectedutil_i(\sigma_1, \state_{p_1}, 2) \leq \max \{
    \lambda_i((f_1+\base)(1-\lambda_s)+ (R-2) \lambda_i (f+\base)) + C,\nonumber\\
    \lambda_i((R-3) \lambda_i (f+\base)+\lambda_i (f_2+\base)) + C\nonumber
    \}
    \end{flalign*}
    \item by $\sigma_2$ we denote strategies where $\player_i$ decides to continue the current chain of the state, adding $txs_u$:  \begin{flalign*}
    \expectedutil_i(\sigma_2, \state_{p_1}, 2) \geq \lambda_i(f+\base+ (R-3) \lambda_i (f+\base) + \lambda_i (f_2+\base)) + C\nonumber
\end{flalign*}
    \item by $\sigma_{f_1}$ we denote a strategy that chooses to fork $txs_{p_1}$ and mine $txs_1$ - \begin{flalign*}
    \expectedutil_i(\sigma_{f_1}, \state_{p_1}, 2) \leq \max\{ \lambda_i((f_1+\base)(1-\lambda_s)+(R-2) \lambda_i (f+\base)) + C,\nonumber\\ 
    \lambda_j( (R-2) \lambda_i (f+\base)+\lambda_i (f_2+\base)) + C\nonumber 
    \} 
\end{flalign*}
    \item by $\sigma_{f_u}$ we denote a strategy that chooses to fork $txs_{p_1}$ and mine $txs_u$
    \begin{flalign*}
    \expectedutil_i(\sigma_{f_u}, \state_{p_1}, 2) \leq \max\{\lambda_i((f+\base)(1-\lambda_s)+ (R-3)\lambda_i(f+\base)+\lambda_i(f_2+\base)) + C, \nonumber\\ 
    \lambda_i(\lambda_i(f_1+\base)+ (R-3)\lambda_i(f+\base)) + C \nonumber
    \}
    \end{flalign*}
\end{itemize}
Similarly, $\sigma_2$ is the dominating strategy for every player $\player_i \neq \player_s$, unless $(f_1+B) (1-\lambda_s) \geq (f+B)$ (which does not hold for $f_1 -f \geq 0.05$, and $\base \approx 6$).

\vspace{0.2cm}\noindent\textbf{Case 2.} Player $\player_i$ (where $i\neq s$) mined $txs_{p_1}$ in the state $\state_{p_1}$.

\vspace{0.2cm}\noindent\textbf{Case 2.A.} The analysis for the player $\player_i$, for some constant $C$:

\begin{itemize}
\item by $\sigma_1$ we denote strategies where the player $\player_i$ decides to continue the current chain of the state, adding $txs_1$:  \begin{flalign*}
    \expectedutil_i(\sigma_1, \state_{p_1}, 2) \leq \max \{
    \lambda_i(\reward_i(\state_{p_1}) +(f_1+\base)(1-\lambda_j)+ (R-2) \lambda_i (f+\base)) + C,\nonumber\\
    \lambda_i(\reward_i(\state_{p_1}) + (R-3) \lambda_i (f+\base)+\lambda_i (f_2+\base)) + C
    \}\nonumber
    \end{flalign*}
    \item by $\sigma_2$ we denote strategies where the $\player_i$ decides to continue the current chain of the state, adding $txs_u$:  \begin{flalign*}
    \expectedutil_i(\sigma_2, \state_{p_1}, 2) \geq \lambda_i(\reward_i(\state_{p_1}) +f+\base+ (R-3) \lambda_i (f+\base) + \lambda_i (f_2+\base)) + C\nonumber
\end{flalign*}
    \item by $\sigma_{f_1}$ we denote a strategy that chooses fork $txs_{p_1}$ and mine $txs_1$ - \begin{flalign*}
    \expectedutil_i(\sigma_{f_1}, \state_{p_1}, 2) \leq \max\{ \lambda_i(\reward_i(\state_{p_1}) + (f_1+\base)+(R-2) \lambda_i (f+\base)-\base) + C,\nonumber\\ 
    \lambda_j(\reward_i(\state_{p_1}) + (R-2) \lambda_i (f+\base)+\lambda_i (f_2+\base)-\base) + C\nonumber 
    \} 
\end{flalign*}
    \item by $\sigma_{f_u}$ we denote a strategy that chooses to fork $txs_{p_1}$ and mine $txs_u$
    \begin{flalign*}
    \expectedutil_i(\sigma_{f_u}, \state_{p_1}, 2) \leq \max\{\lambda_i(\reward_i(\state_{p_1}) + f+\base+ (R-3)\lambda_i(f+\base)+\lambda_i(f_2+\base)-\base)+ C, \nonumber\\ 
    \lambda_i(\reward_i(\state_{p_1}) +\lambda_i(f_1+\base)+ (R-3)\lambda_i(f+\base)-\base) + C \nonumber
    \}
    \end{flalign*}
\end{itemize}

\noindent Similarly to previous case, the strategy $\sigma_2$ dominates all other strategies. Again, the above analysis will be similar, even if some player other than the $\player_i$ successfully starts a fork of $txs_{p_1}$.\\

\vspace{0.2cm}\noindent\textbf{Case 2.B.} Next, we analyze the cases for the player $\player_{s}$ who did not mine $txs_{p_1}$. For some constant $C$:

\begin{itemize}
\item by $\sigma_1$ we denote strategies where player $\player_{s}$ decides to continue the current chain of the state, adding $txs_1$:  \begin{flalign*}
    \expectedutil_{s}(\sigma_1, \state_{p_1}, 2) \leq \max \{
    \lambda_{s}(\reward_s(\state_{p_1}+f_1+\base+ (R-2) \lambda_{s} (f+\base)) + C, \nonumber\\
    \lambda_{s}(\reward_s(\state_{p_1}+ (R-3) \lambda_{s} (f+\base)+\lambda_{s} (f_2+\base)+P) + C
    \}\nonumber
    \end{flalign*}
    \item by $\sigma_2$ we denote strategies where player $\player_{s}$ decides to continue the current chain of the state, adding $txs_u$:  \begin{flalign*}
    \expectedutil_{s}(\sigma_2, \state_{p_1}, 2) \geq \lambda_j(f+\base+ (R-3) \lambda_{s} (f+\base) + \lambda_{s} (f_2+\base)+P) + C\nonumber
\end{flalign*}
    \item by $\sigma_{f_1}$ we denote a strategy that chooses to fork $txs_{p_1}$ and mine $txs_1$: 
    \begin{flalign*}
    \expectedutil_{s}(\sigma_{f_1}, \state_{p_1}, 2) \leq \max\{ \lambda_j(\reward_s(\state_{p_1}+ (f_1+\base)(1-\lambda_i)+(R-2) \lambda_{s} (f+\base) + P) + C,\nonumber\\ 
    \lambda_{s}(\reward_s(\state_{p_1}+(R-2) \lambda_{s} (f+\base)+\lambda_{s} (f_2+\base)+P) + C 
    \} \nonumber
\end{flalign*}
    \item by $\sigma_{f_u}$ we denote a strategy that chooses fork $txs_{p_1}$ and mine $txs_u$:
    \begin{flalign*}
    \expectedutil_{s}(\sigma_{f_u}, \state_{p_1}, 2) \leq \max\{\lambda_j(\reward_s(\state_{p_1}+ (f+\base)(1-\lambda_i)+\\ (R-3)\lambda_{s}(f+\base)+\lambda_{s}(f_2+\base)+P) +C,\nonumber \\ 
    \lambda_{s}(\reward_s(\state_{p_1}+\lambda_{s}(f_1+\base)+ (R-3)\lambda_{s}(f+\base))+C
    \}\nonumber
    \end{flalign*}
\end{itemize}

\noindent The strategy $\sigma_2$ is dominating for the player $\player_s$ who did not mine $txs_{p_1}$, and he will prefer to continue mining $txs_{p_1}$ whenever $f_2-f > \frac{2\overline{f}+2(f_1-f)}{\lambda_s}$. Note that the above analysis will be similar, even if some player other than the $\player_s$ successfully starts a fork of $txs_{p_1}$. 

\vspace{0.2cm}\noindent\textbf{Case 2.C.} Finally, we analyze the case for any player $\player_{i'} \neq \player_s$ different than the player $\player_i$ who mined $txs_{p_1}$. The argument follows similarly as in the \textbf{Case 1.B.} that the strategy that chooses to continue the current chain of the state and adding $txs_u$ is dominating, given that both the player $\player_i$ and the player $\player_s$ will choose to continue mining the transaction $txs_{p_1}$, even if someone else forks it.
\end{proof}

\subsection{Encouraging the Strongest Miner to Use the Penalty Mechanism}
\label{sec:summary:result}
Finally, we observe the benefit that comes from using the penalty mechanism. First for the miner with the strongest mining power $\lambda_s$, we observe that using the self-penalty mechanism and threatening others to mine the transaction $txs_1$, once it appears on the blockchain is beneficial for him whenever $f_2 - f > \frac{2\overline{f}+2(f_1-f)}{\lambda_s}+\overline{f}$. Next, for any miner with a smaller mining power, we show that merely the fact that he is threatened to mine $txs_1$ can force them to skip mining this transaction.
\begin{restatable}[]{thm}{main}
\label{thm:main}
In the game with forks $\game$ that starts with an empty state $\state$, whenever $f_2 - f > \frac{2\overline{f}+2(f_1-f)}{\lambda_s}+\overline{f}$, $\frac{f+\base}{f_1+\base} > 1-\lambda_s^2$, $f_1 > f$, $f_{p_2} > f$, $\lambda_{min} > 0.05^{T/2}$, $f_2-f < 10^{-2}(f+\base)$ and given $Y = \sum_{j=i: \lambda_j > 0.01, f_2-f <\frac{f_1-f}{\lambda_{j}}}^{|\players|} \lambda_j$, it holds that $(1- 1.01(1-(1-Y)^{T/2}) > 0$, the dominating strategy for all players in the game $\game$ is is to mine $txs_{p_1}$ with $P >\lambda_s(f+\base)$ created in the first round by the  strongest player $\player_s$ with the mining power $\lambda_s$, then mine $txs_u$ until round $T$, then $txs_2$, $txs_{p_2}$, and $txs_u$ until the end. 
\end{restatable}
\begin{proof}
By Theorems~\ref{thm:txs1} and ~\ref{thm:txs2}, utility of the player $\player_s$ that chooses to create $txs_{p_1}, txs_{p_2}$ and mine $txs_{p_1}$ in the first round\footnote{In the analysis we omit the strategy where the player $\player_s$ creates $txs_{p_1}, txs_{p_2}$ and does not decide to mine $txs_{p_1}$} (strategy $\sigma_{p}$) is at least:
\begin{flalign*}
    \expectedutil_s(\sigma_{p}, \state, 1) 
\geq -\lambda_{p_1} c_{p_1}\overline{f_{p_1}} + \lambda_s ((m-c_{p_1})\overline{f}+\base) + (\lambda_{p_1}+\lambda_s)F'_2 + (1-\lambda_s-\lambda_{p_1})F'_1,
\end{flalign*}
where $
F'_1= (\lambda_s(T-1)(f+\base)+\lambda_s(R-T)(f+\base)),$
\begin{flalign*}
F'_2 = (\lambda_s((T-1)(f+\base)) + \lambda_s(f_2+\base) - \lambda_{p_2}c_{p_2}\overline{f}_{p_2} + \lambda_s ((m-c_{p_2})\overline{f}+\base) +\\\lambda_i(R-T-2)(f+\base)).
\end{flalign*}

\noindent Recall that we assume that all players with mining power less than $1\%$ have collective power at most $5\%$. As the players with mining power more than $0.01$ will prefer to mine $txs_1$ in the first place when $txs_{p_1}$ is not created, the utility of the player~$\player_s$ that does not decide to create $txs_{p_1}$ (strategy $\sigma_1$) is at most (by Lemma~\ref{lemma:txs1} and the Theorem~\ref{thm:txs2}):
\begin{flalign*}
    \expectedutil_s(\sigma_{1}, \state, 1) \leq ((1-0.05^{T/2})F_1+0.05^{T/2}F_2),
\end{flalign*}
where $
F_1=F'_1+\lambda_s(f_1+B)$, $
F_2=F'_2+\lambda_s(f+B)$.
Further, if $\lambda_{p_1}c_{p_1}\overline{f}_{p_1}+\lambda_s(f+\base)-\lambda_s c_{p_1}\overline{f}+F'_1+(\lambda_{p_1}+\lambda_s)(F'_2-F'_1) > F_1 + 0.05^{T/2}(F_2-F_1),$ then $\expectedutil_s(\sigma_{p}, \state, 1) > \expectedutil_s(\sigma_{1}, \state, 1).$ This condition holds whenever: 
\begin{flalign*}
(\lambda_{p_1}+\lambda_s)[\lambda_s f_2 - \lambda_s f_1] > 0.05^{T/2}[\lambda_s f_2-\lambda_s f_1]+\lambda_{p_1}c_{p_1}\overline{f}_{p_1}+\lambda_s c_{p_1}\overline{f} +\\ (\lambda_{p_1}+\lambda_s)(\lambda_{p_2}c_{p_2}\overline{f}_{p_2}+\lambda_s c_{p_2}\overline{f})+\lambda_s(f_1-f).\end{flalign*} What concludes that the following bribe is enough to encourage the strong miner to wait for $txs_2$:
\begin{flalign*}
f_2-f > \frac{c_{p_1}\overline{f}_{p_1}+c_{p_1}\overline{f}+c_{p_2}\overline{f}_{p_2}+f_1-f}{\lambda_s-0.05^{T/2}}+c_{p_2}\overline{f}.
\end{flalign*}
 
\noindent Now, if the $txs_{p_1}, txs_{p_2}$ are created then every player $\player_i$ other than the player $\player_s$ may mine $txs_{p_1}$, once it is published (strategy $\sigma_{p_1^*}$). When $txs_{p_1}$ is successfully mined in the first round, then all miners will be encouraged to wait until $txs_2$ may be mined after the $T$'th round. $1-\lambda_{p_1}-\lambda_i$ miners may decide to mine $txs_u$ (or $txs_1$) in the first round. In this case, when the $txs_u$ is mined, all other players will be able to mine at least $f+\base$ for the rest of the rounds $\expectedutil_i(\sigma_{p_1^*}, \state, 1) \geq \lambda_i(f_{p_1}+\base)+ (
 \lambda_{p_1}+\lambda_i)F_2 + (1-\lambda_{p_1}-\lambda_i)F$,
where $F_2 = (T-1)\lambda_i(f+B)+\lambda_i(f_2+B)+\lambda_i(f_{p_2}+B)+(R-T-2)\lambda_i(f+B)$ and $F = (R-1)\lambda_i(f+B)$.
 
 On other hand the players may first decide to mine either $txs_{u}$ or $txs_1$ in the first round (strategy $\sigma_{1*}$).
 In the worst case scenario the block with $txs_1$ is not forked. What is more, 
 whenever $1-\lambda_{p_1}$ miners decide to mine $txs_u$ in the first round, then all miners with mining power more than $0.01$ will make an attempt to mine $txs_1$. In conclusion, by Lemma~\ref{lemma:txs1} and the Theorem~\ref{thm:txs2}: 
 \begin{flalign*}
     \expectedutil_i(\sigma_{1*}, \state, 1) \leq \lambda_i(f_{1}+\base)+  \lambda_{p_1}F_2+(1-\lambda_{p_1}-\lambda_1)((1-0.05^{T/2})F+0.05^{T/2}F_2 ),
 \end{flalign*}
where $F$ and $F_2$ are defined as previously.
$\expectedutil_i(\sigma_{p_1*}, \state, 1) > \expectedutil_i(\sigma_{1*}, \state, 1)$ holds whenever:
 \begin{flalign*}\lambda_i(f_{p_1}+B)+\lambda_i (F_2-F) > \lambda_i(f_{1}+B)+(1-\lambda_{p_1}) 0.05^{T/2} (F_2-F)
 \end{flalign*}
Which holds for any $f_{p_1} \geq f_1$ and $\lambda_i > 0.05^{T/2}$.

\noindent Now, by setting $c_{p_1} = 1, c_{p_2} = 1$, $\overline{f_{p_1}}=f_1-f$, $\overline{f_{p_2}} = \overline{f}$, we get a condition $f_2-f > \frac{2\overline{f}+2(f_1-f)}{\lambda_s-0.05^{T/2}}+\overline{f}$,
what for $\lambda_s \approx 20\%$ and sufficiently large $T/2$ concludes
$f_2 - f \gtrsim
 \frac{2\overline{f}+2(f_1-f)}{\lambda_s}+\overline{f}$.
\end{proof}
\section{Example Evaluation}
Using the real-world data analysis of Bitcoin fees and hashpower distribution in major PoW blockchains (see Appendix~\ref{sec:data}), we visualize the improvement our bound $f_2 - f > \frac{2\overline{f}+2(f_1-f)}{\lambda_s}+\overline{f}$ brings compared to the previous result from~\cite{zeta}, namely $f_2 - f \geq \frac{f_1-f}{\lambda_{min}}$. 
Additionally, the Theorem~\ref{thm:main} requires that  $f_2-f < 10^{-2}(f+\base)$ and there exists a player $\player_j$ with mining power $\lambda_j > 0.01$ for which $f_2-f < \frac{f_1-f}{\lambda_j}$.

For example, let us assume that $f_1 - f \approx \overline{f}$, and set $T > 110$. Now, since $\lambda_{min}$ can be fairly estimated to be $\lambda_{min} < 10^{-12}$, we can see that the attack without forking threats could cost in practice around $10^{12} \overline{f}$. On the other hand, the new bound requires only $f_2 - f > \frac{2\overline{f}+2(f_1-f)}{\lambda_s}+\overline{f}$, for $\lambda_s \approx 0.2$, this costs around $f_2 - f >  21 \overline{f}$. The only condition left is that for some miner with $\lambda_j > 0.01$, the following condition must hold $f_2-f<\frac{f_1-f}{\lambda_{j}}$, but the data shows that miners that control approximately $1.5 \% - 2 \%$ of the total mining power usually exist, thus for a miner with mining power $1\% < \lambda_j < 2\%$ it holds that $f_2-f< 50 \cdot \overline{f}$. In summary, if we take any $f_2$ that is larger than $f$ by $21$ up to $50$ times, then the default strategy for all miners is to wait for the bribing transaction.

\section{Related Work and Countermeasures}
In the landscape of constructing financially stable systems on blockchain~\cite{miller2017zerocollateral,financiallystable2}, our work falls into the class of incentive manipulation attacks which have been widely applied to undermine blockchain's security assumptions~\cite{mirkin2020}.
To the best of our knowledge, we are the first to combine  feather forking attacks~\cite{perspectives} with timelock bribing attacks on payment channel networks and to achieve a bribing cost that is approximately only constant times larger than the cost of an average transaction fee. 

Incentive manipulation attacks on timelocked puzzles were introduced with the so-called timelock bribing attack~\cite{timelocked}. Later, Avarikioti et al.~\cite{zeta} applied timelock bribing attacks in payment channel networks, such as the Lightning Network and Duplex Micropayment Channels, and proposed countermeasures. 
Our work extends \cite{zeta}, modifying the timelock bribing attack for payment channels to facilitate a miner bribing strategy that incorporates feather forking. As a result, our work reduces the cost of bribing attacks significantly in comparison to \cite{zeta}, i.e., the cost is now inversely proportional to the mining power of the largest miner instead of the smallest miner which results in at least 1000 times smaller bribes. Our model is similar to the one in~\cite{karakostas2024blockchain} that introduced forks, but we were able to craft reasonable assumptions for the PoW blockchains which secured a \emph{unique} NE solution. In particular, we restrict the strategy of the miner by forbidding him to conduct shallow forks and allowing him to fork only in a case when the player strictly profits from conducting the fork action (compare Sections~\ref{sec:threat}, \ref{sec:shallow:forks}).

The bribing strategies for the payment channels are similar in their nature to the bribing strategies for the HTLC mechanism.
Perhaps the closest to our work is~\cite{temporary}, where the authors introduced a way to bribe HTLCs, leveraging the power of smart contracts and feather forking. The cost of the attack in~\cite{temporary} is, however, proportional to the sum of the fees ($\gtrsim \sum_{i=1}^{T} f\cdot \lambda_{max}$) of all blocks before the deadline $T$. In contrast, we achieve a cost proportional to the cost of fees of a single block ($\gtrsim \frac{f_1-f}{\lambda_S}$).

Furthermore, MAD-HTLC~\cite{madhtlc} underlined the vulnerability of HTLCs to bribing attacks, achieving the same attack cost as \cite{zeta}, specifically $\approx \frac{f_1-f}{\lambda_{min}}$. MAD-HTLC presupposes that the minimum fraction of mining power controlled by a single user, $\lambda_{min}$, is at least $0.01$, to achieve low bribing costs. This is, however, an impractical assumption, as the data in Section~\ref{sec:experimental} show that $\lambda_{min}$ can be reasonably estimated to be less than $10^{-12}$, making the bribing attack exceedingly expensive. The reduction of the bribing costs \attack achieves in comparison to MAD-HTLC is similar to that of \cite{zeta} analyzed above.

MAD-HTLC additionally proposed a countermeasure for bribing attacks where miners are allowed to claim the locked coins in the HLTC in case a party misbehaves, similar to \cite{zeta}. Later, He-HTLC~\cite{helium} pointed out that MAD-HTLC is susceptible to counter-bribing attacks. 
In particular, one party may (proactively) collude with the miners to cooperatively steal the coins of the counterparty in the MAD-HTLC construction. 
He-HTLC also proposed a modification on MAD-HTLC to mitigate the counter-bribing attack: now the coins are partially burned in case of fraud instead of being fully awarded to the miners.
Recently, Rapidash~\cite{rapidash} revisited the counter-bribing attack and proposed yet another improvement on He-HTLC. 
These works are orthogonal to ours as the proposed attacks apply only to the specific MAD-HTLC construction and not to Lighting Channels that are the focus of this work.
Furthermore, our focus is not on designing countermeasures against timelocked bribing attacks. 
Instead, we demonstrate how employing feather forking can make timelocked bribing attacks very cheap for the attacker, therefore highlighting the need for robust mitigating strategies.  


Nonetheless, it is crucial to acknowledge that the previously mentioned countermeasures can be used to defend against \attack -- inheriting their respective vulnerabilities. For example, one can employ the mitigation technique for timelocked bribing on the Lighting Network proposed by Avarikioti et al.~\cite{zeta}. In our model, this countermeasure ensures that announcing $txs_2$ also involves revealing a secret that anyone can use to claim the money before time $T$. This implies that if $txs_2$ is announced in the mining pool before time $T$, all the money to be collected only after time $T$ can be immediately claimed by another party.
We assert, without proof, that the same countermeasure mechanism remains effective even in a model that considers forks. Intuitively, the ``strong'' miner in our analysis does not benefit from waiting for the bribing transaction if it is not announced, thus preventing the creation of the self-penalty transaction. Conversely, if the transaction is announced and the secret is revealed, any (winning) miner could claim the reward.

\section{Experimental Analysis}\label{sec:data}
In this section, we conduct the analysis of the fee structures on the real-world data. 

\subsection{Fee Structure}\label{sec:fee}
We analyze the real-world data on fees earned by miners on Bitcoin. In particular, we explore the Bitcoin historical data from 2022 (\cite{btcexplorer}). In this period, the block reward was the sum of the block reward base, 6.25 BTC, and the fees for including transactions. The weekly average sum of the block fees ranged between 0.05 and 0.225 BTC as shown in Figure~\ref{fig:avfee}. The average fee of a single transaction fluctuated proportionally between $0.00004$ and $0.00016$ BTC as shown in Figure~\ref{fig:avtxfee}.
We can thus fairly assume that any block typically contains on average $1500$ transactions, the sum of fees of an average block of transactions is around $0.15$ BTC and the average fee of a single transaction is approximately $10^{-4}$ BTC.

\begin{figure}
\centering
\includegraphics[width=\textwidth]
{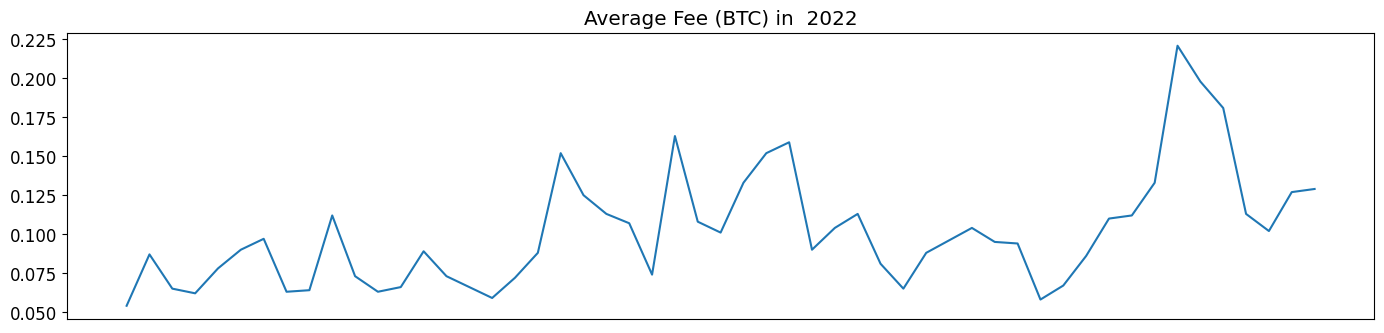}
\caption{Average amount of fees aggregated in a block for each week in 2022 (\cite{btcexplorer})}
\label{fig:avfee}
\end{figure}

\begin{figure}
\centering
\includegraphics[width=\textwidth]
{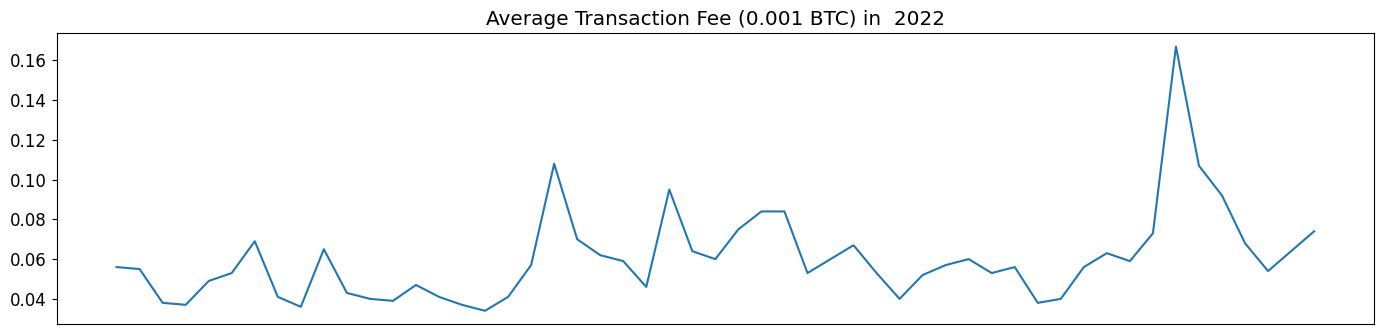}
\caption{Average transaction fee of a single transaction for each week in 2022 (\cite{btcexplorer})}
\label{fig:avtxfee}
\end{figure}

\subsection{Hash Power Analysis}\label{sec:experimental}
\begin{figure}[htb!]
    \centering
\makebox[\textwidth][c]{\includegraphics[width = 16cm]{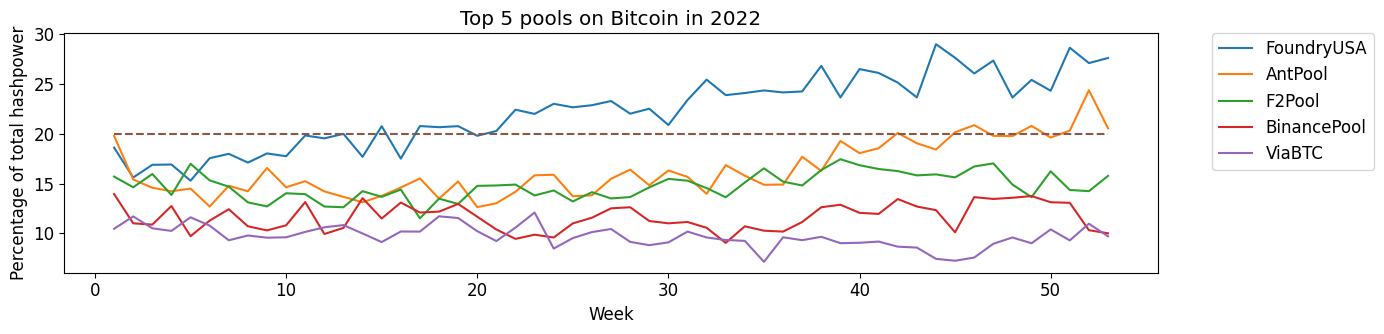}}

\caption{Average mining hashrate week by week of the most significant Bitcoin miners in 2022~\cite{btcexplorer}}
    \label{fig:average:hashrate}
\end{figure}
In this section we present the analysis of hashpower distribution of mining pools in selected popular Proof-of-Work blockchains: Bitcoin, Litecoin and Ethereum Classic. Even though all three blockchains are built upon the PoW concept, they differ in the cryptographic primitive (hash function) that they use: Bitcoin uses SHA-256, Litecoin uses Scrypt and Ethereum Classic uses ETChash. As the speed of Bitcoin mining relies solely on the amount of hashes that a device is able to compute in a time unit, miners developed a market for Application-Specific Integrated Circuits (so-called ASICs) that are designed to calculate hashes specifically for the purpose of Bitcoin mining. This can be considered a rather undesirable phenomenon, as it causes significant environmental losses, thus both Scrypt and ETCHash were designed to lower the advantage of mining blocks with ASICs. Both Scrypt and ETCHash, in contrast to SHA-256, require hardware that has relatively fast access to large amounts of memory, that makes current Bitcoin ASICs unusable in the context of Scrypt and ETCHash mining.

\begin{table}[]
\centering
\begin{tabular}{|l|l|}
\hline
\textbf{Device}             & \textbf{Hashrate (Hashes/second)}                    \\ \hline
Nvidia GeForce RTX 4090     & $2.75 \cdot 10^8$      \\ \hline
Nvidia GeForce RTX 3090     & $9.4 \cdot 10^7$      \\ \hline
Bitmain Antminer S19 XP Hyd & $2.55 \cdot 10^{14}$ \\ \hline
Baikal BK-G28               & $2.8 \cdot 10^{10}$  \\ \hline
\end{tabular}
\caption{Hash rates of commonly used Bitcoin mining devices~\cite{asicminer}}
\label{tab:mining:devices}
\end{table}

\subsubsection{The weakest miner} In the previous work~\cite{zeta}, the authors assumed that the fraction of the mining power held by the weakest miners is approximately $10^{-3}$. The data depicted in the Table~\ref{tab:mining:devices} and in the Figure $2(b)$ shows that it is a highly unrealistic assumption for Bitcoin. A very strong GPU which is often used by individual miners for mining has a hashrate of around $3\cdot 10^8$ (Table~\ref{tab:mining:devices}), whereas the collective mining power in Bitcoin can be estimated to be $2 \cdot 10^{20}$ (Figure $2(b)$). In conclusion, we cannot expect that every individual miner has more than $10^{-12}$ probability of mining a block. 

\subsubsection{Fraction of the weak miners} In this work we assume that there exists a single player that has mining power slightly more than $1\%$, but the collective mining power of weak miners that have mining power less than some $1\%$ is not more than $5\%$. This phenomenon is quite realistic, as observed in Table~\ref{table:pools}.

\subsubsection{Strong miners} On the other side of the spectrum, top mining pools continuously maintain over $20\%$ of total hashpower. To calculate the share of the hashpower of a big mining pool in a given period, we count the number of the blocks mined by it. Then we divide this number by the total number of mined blocks in the period.  Figure~\ref{fig:collective}(a) and Figure~\ref{fig:average:hashrate} show that, in the most popular blockchains, most of the time, one of the miners (mining pools) has the collective power higher than $20\%$.

    

\begin{table}[]
\centering
\begin{tabular}{|l|l|l|}
\hline
\textbf{Pool}             & \textbf{Number of mined blocks}   & \textbf{Percentage of total mined blocks}                 \\ \hline

Foundry USA    & $11851$   &  $22.281\%$  \\ \hline
AntPool    & $8670$   &  $16.301\%$  \\ \hline
F2Pool    & $7848$   &  $14.755\%$  \\ \hline
Binance Pool    & $6164$   &  $11.589\%$  \\ \hline

ViaBTC    & $5162$   &  $9.705\%$  \\ \hline
Poolin    & $4438$   &  $8.344\%$  \\ \hline
BTC.com    & $2252$   &  $4.234\%$  \\ \hline
SlushPool    & $2142$   &  $4.027\%$  \\ \hline
Luxor    & $1362$   &  $2.561\%$  \\ \hline

SBI Crypto    & $963$   &  $1.811\%$  \\ \hline
Braiins Pool    & $715$   &  $1.344\%$  \\ \hline
unknown    & $698$   &  $1.312\%$  \\ \hline
MARA Pool    & $307$   &  $0.577\%$  \\ \hline
Terra Pool    & $149$   &  $0.280\%$  \\ \hline

ULTIMUS POOL     & $147$   &  $0.276\%$  \\ \hline
KuCoinPool    & $101$   &  $0.190\%$  \\ \hline
EMCDPool   & $80$   &  $0.150\%$  \\ \hline

PEGA Pool     & $76$   &  $0.143\%$  \\ \hline
\end{tabular}
\caption{The scale of mining hashrates of the most significant mining pools and independent miners of Bitcoin in 2022. Data sources:~\cite{btcexplorer}}
\label{table:pools}
\end{table}

\begin{figure}
\centering
\begin{subfigure}{.5\textwidth}
  \centering
  \includegraphics[width=\textwidth]
  {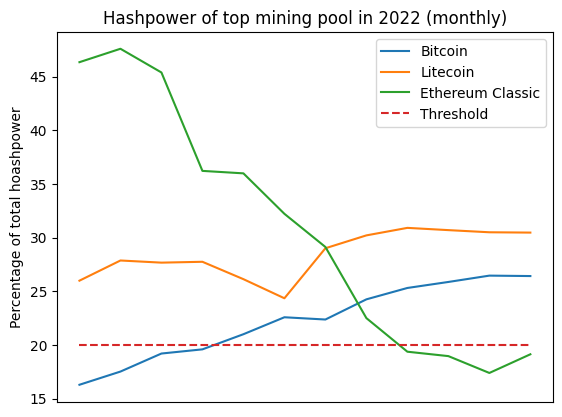}
  \caption{}
  \label{fig:sub1}
\end{subfigure}%
\begin{subfigure}{.5\textwidth}
  \centering
  \includegraphics[width=\textwidth]
  {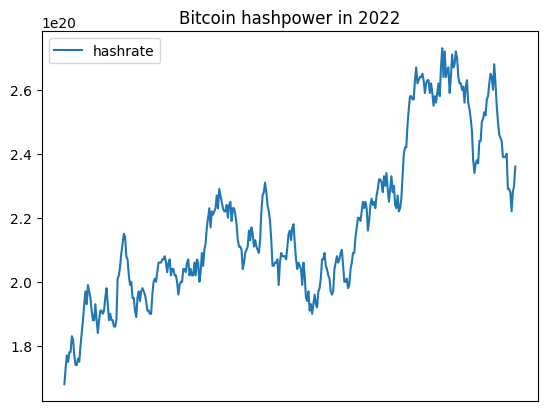}
  \caption{}
  \label{fig:sub2}
\end{subfigure}
\caption{\textbf{(a)} Collective power of top mining pools in 2022 for selected Blockchains, monthly data~\cite{btcexplorer},~\cite{ltc},~\cite{etc}. \textbf{(b)} Collective number of hashes per second calculated by all miners in 2022 [$10^{20}$ hashes/second], weekly data~\cite{hashpower}}
\label{fig:collective}
\end{figure}

\section{Conclusions and Future Work}
In conclusion, our work reexamines the vulnerability of PCNs to bribing attacks and introduces a modified attack leveraging the threat of forking. We introduce a formal model of a mining game with forking extending the conditionally timelocked games introduced by Avarikioti et al.~\cite{zeta}. In particular, in our extended model, miners not only choose which transactions to mine in each round but also decide whether to continue existing chains or initiate forks. In this model, we demonstrate that the cost of the bribing attack can be significantly reduced compared to the previous analysis. In more detail, it can be reduced from $\frac{f_1-f}{\lambda_{min}}$ to approximately $\frac{2\overline{f}+2(f_1-f)}{\lambda_s}$, where $\overline{f}$ represents the cost of an average fee for a single transaction and $\lambda_s$ denotes the reduction factor compared to significantly smaller $\lambda_{min}$ calculated in previous work~\cite{zeta}. To validate our findings, we empirically analyze the historical data of real-world blockchain implementations. This analysis confirms that staging a bribing attack on a PCN is significantly less costly (approximately 125\$) than considered previously. 

The results of our study have implications for the design and implementation of PCNs, as well as for the broader applications of timelocked contracts, e.g., atomic swaps. Our findings underscore the need for proactive measures to mitigate the risk of bribing attacks. 

Possible avenues for future research include exploring whether our penalty mechanism implementation can be implemented without the honest majority assumption or whether our results still hold in the presence of more general abandon rules. Another interesting question is whether our results extend in a Proof-of-Stake setting.

\section*{Acknowledgments}
We thank Paul Harrenstein for his help in defining the model presented in this work

This work was supported by the Austrian Science Fund (FWF) through the SFB SpyCode project F8512-N, the project CoRaF (grant agreement ESP 68-N), and by the WWTF through the project 10.47379/ICT22045.

This result is part of a project that received funding from the European Research Council (ERC) under the European Union's \emph{Horizon 2020} and \emph{Horizon Europe} research and innovation programs (grants PROCONTRA-885666 and CRYPTOLAYER-101044770). This work was also partly supported by the National Science Centre, Poland, under research project No.~46339.

\section{Notation Summary}\label{app:notation}
The Table~\ref{tab:notation}
 contains summary of the notation used in the proofs of the theorems.
\begin{table}[h!]
\begin{tabular}{|c|l|}
\hline
\textbf{Symbol}                                               & \textbf{Meaning}                                                                                                                                                       \\ \hline
$\players = \{\player_1, \ldots, \player_{|\players|}\}$                                        & Set of players (miners).                                                                                                                        \\ \hline
$\boldsymbol{\lambda} = (\lambda_1, \ldots, \lambda_{|\players|})$       & Tuple of mining powers of the players.     
\\ \hline
$\game$ & Conditionally timelocked game with forks. \\ \hline
$\state = \{S_1, \ldots, S_{|\state|}\}$                        & Current state in the game consisting of a set of chains.                                                                                                                          \\ \hline

$R,r$ & Number of rounds in the game, current round in the game. \\ \hline
$tx_i$                                                 & Single transaction that may be included in a transactions set.                                                                                    \\ \hline
$txs_i$                                                 & Transactions set that may be mined to create a single block.                                                                          \\ \hline
$\base, P$                                                       & The base of a transaction reward, the punishment added to \\&(collected from) a block.                                                                                                                                      \\ \hline
$f, \overline{f}$                                             & Fees of an average transaction set, fee of an average transaction.                                                                                                     \\ \hline
$f_i, \overline{f}_i$                                         & Fees of a transaction set with a transaction $tx_i$, \\&fee of the transaction named $i$.                                                                                       \\ \hline

$tx_u,tx_1,tx_2,tx_b,tx_{p_1},tx_{p_2}$                                         & Transactions used to construct $\attack$ \\ &(see the paragraph "Types of transaction sets.").  \\ \hline
$txs_u,txs_1,txs_2,txs_{p_1},txs_{p_2}$                                         & Transaction sets that include transactions \\& used to construct $\attack$ (see the paragraph \\ &"Types of transaction sets.").  \\ \hline

$\reward(txs_i)$                                              & Total reward from a transaction set.                                                                                                                                   \\ \hline
$\reward_i(\state)$                                           & The reward collected by a player $i$ in the state $\state$.                                                                                                             \\ \hline
$\sigma = (\sigma_1, \ldots, \sigma_n)$                       & The strategy of the players.                                                                                                                                           \\ \hline
$\expectedutil_i(\sigma), \expectedutil_i(\sigma, \state, r)$ & Utility of a player $i$, when strategy $\sigma$ is played,  utility of the \\&player $i$ in a subgame starting in state $\state$, round $r$ and strategy $\sigma$. \\ \hline
$C$                                                           & A constant describing the utility of players not considered\\& in current analysis (see the paragraph "About the proofs.").                                        \\ \hline
\end{tabular}
\caption{Notation summary}\label{tab:notation}
\end{table}

\bibliography{main}

\end{document}